\documentclass[a4paper]{article}

\usepackage{amsmath, amssymb, amsthm}
\usepackage{graphicx}
\usepackage{hyperref}
\usepackage{thmtools}
\usepackage{orcidlink}
\usepackage{authblk}

\newtheorem{theorem}{Theorem}
\newtheorem{claim}{Claim}[theorem]
\newtheorem{corollary}[theorem]{Corollary}

\newtheorem{proposition}[theorem]{Proposition}
\newtheorem{openquestion}{Open Question}
\DeclareMathOperator{\diam}{diam}
\newcommand{\dmaxproblem}{\textsc{$d_{\max}(\mathcal{G})$-Diameter}}

\usepackage{regexpatch}
\makeatletter
\xpatchcmd\thmt@restatable{%
\csname #2\@xa\endcsname\ifx\@nx#1\@nx\else[{#1}]\fi
}{%
\ifthmt@thisistheone
\csname #2\@xa\endcsname\ifx\@nx#1\@nx\else[{#1}]\fi
\else
\csname #2\@xa\endcsname[{Restated}]
\fi}{}{}
\makeatother

\title{The Complexity of Diameter on $H$-free graphs\thanks{A peer-reviewed extended abstract of this work appeared in the proceedings of WG 2024~\cite{OostveenPL25Diameter}. A full peer-reviewed version is to appear in SIDMA.}}
\author[1]{Jelle J. Oostveen\ \orcidlink{0009-0009-4419-3143}\thanks{J.J. Oostveen is supported by the NWO grant OCENW.KLEIN.114 (PACAN).}}
\author[2]{Daniël Paulusma \orcidlink{0000-0001-5945-9287}}
\author[1]{Erik Jan van Leeuwen \orcidlink{0000-0001-5240-7257}}

\affil[1]{Dept.\ Information and Computing Sciences, Utrecht University, The Netherlands\\ \texttt{j.j.oostveen@uu.nl}, \texttt{e.j.vanleeuwen@uu.nl}}
\affil[2]{Dept.\ Computer Science, Durham University, United Kingdom\\ \texttt{daniel.paulusma@durham.ac.uk}}

\date{}

\begin{document}
	\maketitle
	
	\begin{abstract}
		The intensively studied {\sc Diameter} problem is to find the diameter of a given connected graph. We investigate, for the first time in a structured manner, the complexity of {\sc Diameter} for $H$-free graphs, that is, graphs that do not contain a fixed graph $H$ as an induced subgraph. We first show that if $H$ is not a linear forest with small components, then {\sc Diameter} cannot be solved in subquadratic time for $H$-free graphs under SETH\@. For some small linear forests, we do show linear-time algorithms for solving {\sc Diameter}. For other linear forests $H$, we make progress towards linear-time algorithms by considering specific diameter values. If $H$ is a linear forest, the maximum value of the diameter of any graph in a connected $H$-free graph class is some constant $d_{\max}$ dependent only on $H$. We give linear-time algorithms for deciding if a connected $H$-free graph has diameter $d_{\max}$, for several linear forests $H$. In contrast, for one such linear forest~$H$, {\sc Diameter} cannot be solved in subquadratic time for $H$-free graphs under SETH\@. Moreover, we even show that, for several other linear forests $H$, 
		one cannot decide in subquadratic time if a connected $H$-free graph has diameter $d_{\max}$ under SETH\@.
	\end{abstract}

	\section{Introduction}
	
	The {\sc Diameter} problem asks to find the diameter of an undirected, unweighted graph, that is, the longest of the shortest paths between all pairs of nodes. We denote for a graph $G = (V,E)$ the quantities $n = |V|$ and $m = |E|$. A trivial algorithm executes a Breadth First Search (BFS) from every node in the graph, and has a running time of $O(nm)$. The best known matrix multiplication-based algorithms achieve a running time of $\tilde{O}(n^{\omega})$~\cite{CyganGS15,ShoshanZ99,Zwick02} to find the diameter of a graph, where $\tilde{O}$ hides logarithmic factors and $\omega$ is the matrix multiplication constant, with current known value $\omega < 2.371866$~\cite{DuanWZ23}. A search for improvement led to a hardness result; that, under the Strong Exponential Time Hypothesis (SETH), one cannot decide between diameter 2 or 3 on (sparse) split graphs in $O(n^{2-\epsilon})$ time, for any $\epsilon > 0$~\cite{RodittyWilliams13}. SETH is a hypothesis that states that {\sc Satisfiability} cannot be solved in $2^{(1-\epsilon)n}$ time, for any $\epsilon > 0$, where $n$ is the number of variables~\cite{ImpagliazzoP01,ImpagliazzoPZ01}.
	On other simple graph classes like constant degree graphs, truly subquadratic time algorithms for {\sc Diameter} are also ruled out under SETH~\cite{EvaldDahlgaard16}.
	Directed versions of the {\sc Diameter} problem admit similar barriers under SETH~\cite{AbboudWW16}. No clear bound is known for {\sc Diameter} on dense graphs, and no such lower bound can be derived when $\omega=2$, but we do know that there is a subcubic equivalence between {\sc Diameter} and computing the reach centrality of a graph, that is, a truly subcubic algorithm for one implies such an algorithm for the other and vice versa~\cite{AbboudGW23}.
	
	Given that the hardness results are based on long-standing conjectures, it is natural to approach diameter computation and other similar problems on restricted graph classes. Related literature also concerns computation of eccentricities, as computing the diameter of a graph is equivalent to computing the largest eccentricity over all vertices. A conceptually simple algorithm called LexBFS can solve {\sc Diameter} in $O(n+m)$ time for distance-hereditary chordal graphs and interval graphs~\cite{dragan1997lexbfs}. Distance-hereditary graphs have been studied separately, and admit linear-time algorithms of all eccentricities~\cite{Dragan94,DraganG20,DraganN00}. Interval graphs admit computation of the eccentricity of the center of the graph in linear time, next to linear-time diameter computation~\cite{Olariu90a}. Subquadratic algorithms for {\sc Diameter} and computing eccentricities have been studied for more graph classes, including asteroidal triple-free graphs~\cite{ducoffe2022diameter}, directed path graphs~\cite{CorneilDHP01}, strongly chordal graphs~\cite{dragan1989centers}, dually chordal graphs~\cite{BrandstadtCD98,Dragan93}, Helly graphs and graphs of bounded Helly number~\cite{DraganDG21,Ducoffe23Helly,DucoffeD21}, $\alpha_i$-metric graphs~\cite{DraganD23MetricGraphs}, retracts~\cite{Ducoffe21BeyondHelly}, $\delta$-hyperbolic graphs~\cite{ChepoiDEHV08,ChepoiDHVA19,DraganG20,DraganHV18}, planar graphs~\cite{AbboudMW23,Cabello19,GawrychowskiKMS21}, and outerplanar graphs~\cite{FarleyP80}. {\sc Diameter} was also studied from the parameterized perspective, see e.g.~\cite{AbboudWW16,bringmann2020multivariate,CoudertDP19,ducoffe2022optimal,DucoffeHV22}, and a large body of work exists on approximation algorithms, see e.g.~\cite{AbboudWW16,BackursRSWW18,ChechikLRSTW14,CorneilDHP01,CorneilDK03,DraganD23MetricGraphs,RodittyWilliams13,WeimannY16}.

	For graph classes with forbidden patterns, Ducoffe et al.~\cite{DucoffeHV22} show subquadratic-time computation of {\sc Diameter} for $\mathcal{H}$-minor-free graphs, where the precise exponent depends on $\mathcal{H}$, improved upon by Le and Wulff-Nilsen~\cite{LeW24}. Johnson et al.~\cite{FrameworkPaper} showed for $\mathcal{H}$-subgraph-free graphs that there is a dichotomy for {\sc Diameter} between almost-linear time solvability and quadratic-time conditional lower bounds depending on the family $\mathcal{H}$. Also note that many of the studied graph classes listed above can be characterized as $\mathcal{H}$-free graphs for a family of graphs $\mathcal{H}$, but for each one, $\mathcal{H}$ has size 2 or larger. As far as we are aware, a structured study into forbidden (monogenic) induced patterns is absent in the literature.
	
	\paragraph{Our Contributions.} We initiate a structured study into diameter computation on $H$-free graphs, where $H$ is a single graph. Recall that a graph is $H$-free if it does not contain $H$ as an induced subgraph. The question we consider is:
	
	\medskip\noindent
	\emph{For which $H$-free graph classes $\mathcal{G}$ can the diameter of an $n$-vertex graph $G\in \mathcal{G}$ be computed in time $O(n^{2-\epsilon})$?}

	\medskip\noindent
	Our first result analyses existing lower bounds to find hardness for $H$-free graph classes. Here, $P_t$ denotes the path on $t$ vertices, and $sP_t$ denotes the disjoint union of $s$ copies of a $P_t$. For graphs $G,H$ let $G+H$ denote their disjoint union. Recall that a \emph{linear forest} is a disjoint union of one or more paths.

	\begin{restatable}{theorem}{HardnessGeneral}\label{thm:hardnessgeneral}
		Let $H$ be a graph that contains an induced~$2P_2$ or is not a linear forest. Under SETH, {\sc Diameter} on $H$-free graphs cannot be solved in $O(n^{2-\epsilon})$ time for any $\epsilon > 0$.
	\end{restatable}
	
	Theorem~\ref{thm:hardnessgeneral} shows the most prominent gap to be for graph classes which exclude a small linear forest.
	Ducoffe~\cite{ducoffe2022diameter} proved a hardness result for {\sc Diameter} on AT-free graphs that holds under the hypothesis that the currently asymptotically fastest (combinatorial) algorithms for finding simplicial vertices (vertices with a complete neighbourhood) are optimal, which we refer to as the \emph{Simplicial Vertex Hypothesis}. Because the vertex set of the graph in Ducoffe's hardness construction can be partitioned into four cliques, the result of~\cite{ducoffe2022diameter} can be formulated as follows.

	\begin{theorem}[\cite{ducoffe2022diameter}]\label{thm:simp}
		For any fixed $k\geq 5$, under the Simplicial Vertex Hypothesis, there does not exist a combinatorial algorithm for {\sc Diameter} on $kP_1$-free graphs that runs in $O(m^{3/2 - \epsilon})$ time for any $\epsilon>0$.
	\end{theorem}

	To complement the hardness results, we show linear-time algorithms for several classes of $H$-free graphs for which $H$ is a small linear forest. Here, $H \subseteq_i H'$ denotes that $H$ is an induced subgraph of $H'$.

	\begin{restatable}{theorem}{DiamAlgGeneralOverview}\label{thm:diamalggeneraloverview}%
		Let $H \subseteq_i P_2 + 2P_1$, $P_3 + P_1$, or $P_4$. Then {\sc Diameter} on $H$-free graphs can be solved in $O(n+m)$ time.
	\end{restatable}

	We achieve Theorem~\ref{thm:diamalggeneraloverview} by careful structural analysis of the graph class and then show that a constant number of Breadth First Searches suffice algorithmically. Note that a running time of $O(n+m)$ clearly beats the naive algorithm of $O(nm)$ time and the matrix multiplication algorithms of $\tilde{O}(n^\omega)$ time, but also rules out any quadratic lower bound in $n$, as the classes of graphs contain abitrarily large families of sparse graphs, e.g.\ stars.
	
	Combining Theorems~\ref{thm:hardnessgeneral},~\ref{thm:simp} and~\ref{thm:diamalggeneraloverview}, the only open cases for the complexity of {\sc Diameter} on $H$-free graphs are $H=4P_1$, $H=P_2+3P_1$, $H = P_3+2P_1$, $H=P_4+2P_1$, and $H = P_4+P_1$. 
	The smallest graph $H$ that is an open case is that of $H = 4P_1$. As a `hardness' result for $4P_1$-free graphs, one could try to take the split graph construction of Roditty and Williams~\cite{RodittyWilliams13} (see Theorem~\ref{thm:RodWilHardness} in Appendix~\ref{appendix}), and add edges to make the graph consist of three cliques (as in~\cite{CorneilDHP01}). Conceptually, this would seem to work: the diameter distinction is still 2 or 3 and translates to a SAT positive or negative answer. However, this approach fails due to the quantity of edges one adds to the graph. The lower bound shows that no $O(n^{2-\epsilon})$ time algorithm may exist for this new instance, which is now a relatively empty lower bound: the graph has a quadratic number of edges, so this lower bound does not even rule out an $O(n+m)$ time algorithm. The density of graphs matters in relation to lower bounds, and seems to provide a barrier to finding a lower bound that rules out a linear-time algorithm.
	
	However, if we adopt the perspective from the other side, a linear-time algorithm for $4P_1$-free graphs would still be surprising. Indeed, such an algorithm that can decide between diameter 2 or 3 on the three-clique instance described earlier implies an algorithm for {\sc Orthogonal Vectors} in time $O(n^2 + d^2)$, where $d$ is the dimension of the vectors and $n$ the size of the vector sets (see a discussion in Appendix~\ref{appendix}). Although lower bounds do not rule out this possibility, such a result would be highly non-trivial, as the best known algorithms for {\sc Orthogonal Vectors} do not achieve this running time for all $d$~\cite{AbboudWY15,ChanW21}. Any linear-time algorithm would even beat the best known matrix-multiplication algorithms of $\tilde{O}(n^{\omega})$ time, even if $\omega = 2$. It thus seems we are at an impasse to find or exclude a linear-time algorithm for computing the diameter of $4P_1$-free graphs.
	
	However, as it turns out, we \emph{can} decide in linear time whether the diameter of a $4P_1$-free graph is exactly 5. Our approach avoids the above barriers by focusing on specific diameter values instead of deciding on the diameter of a graph completely.

	In general, for a graph class $\mathcal{G}$, we call $d_{\max}(\mathcal{G})$ the maximum diameter that any graph in $\mathcal{G}$ can have; formally $d_{\max}(\mathcal{G}) = \sup_{G\in \mathcal{G}}\diam(G)$. We omit $\mathcal{G}$ when it is clear from context. In particular, for $4P_1$-free graphs, $d_{\max}$ is equal to 5. We define the \dmaxproblem{} problem as deciding for a graph $G\in \mathcal{G}$ whether it holds that $\diam(G) = d_{\max}(\mathcal{G})$. The research question we investigate is:
	
	\medskip\noindent
	\emph{For which $H$-free graph classes $\mathcal{G}$ can we solve \dmaxproblem{} in linear time?}
	
	\medskip\noindent
	For some classes $\mathcal{G}$, it is easy to see $d_{\max}(\mathcal{G})$ is bounded. For instance, the class of cliques has $d_{\max} = 1$. Any graph class that contains paths of arbitrary length has $d_{\max} = \infty$. For deciding whether the diameter of a graph is equal to $d_{\max}$, only classes with bounded $d_{\max}$ value are interesting to consider. It turns out that for classes of connected $H$-free graphs, $d_{\max}$ is bounded exactly when $H$ is a linear forest (see Theorem~\ref{thm:HLongestPath} in the preliminaries).

	Our contributions with respect to the \dmaxproblem{} problem are twofold. 
	Firstly, we find several examples of $H$-free classes $\mathcal{G}$ where $H$ is a linear forest of more than one path where we solve \dmaxproblem{} in linear time. Note that $d_{\max}$ can differ vastly for classes where $H$ is a linear forest, depending on $H$.

	\begin{restatable}{theorem}{AlgOverview}\label{thm:AlgOverview}%
		Let $H=2P_2+P_1$ or $H\subseteq_i P_2 + 3P_1$, $P_3+2P_1$, or $P_4 + P_1$, and let $\mathcal{G}$ be the class of $H$-free graphs. Then \dmaxproblem{} can be solved in $O(n+m)$ time.
	\end{restatable}

	Note that in particular, Theorem~\ref{thm:AlgOverview} shows that we can decide whether $\diam(G) = d_{\max}(\mathcal{G})$ for all previously stated open cases for \textsc{Diameter} computation on $H$-free graphs, except for the case of $H=P_4+2P_1$.

	Secondly, we extend known hardness constructions to hold for the \dmaxproblem{} problem for certain $H$-free graph classes $\mathcal{G}$. Note that one needs different hardness proofs for different $H, H'$, even if $H \subseteq_i H'$, because $d_{\max}$ can differ for both classes.
	\begin{restatable}{theorem}{OddPtHard}\label{thm:Ptodd}%
		Let $H=2P_2$ or $H=P_t$ for some odd $t\geq 5$, and let $\mathcal{G}$ be the class of $H$-free graphs. Under SETH, it is not possible to solve \dmaxproblem{} in time $O(n^{2-\epsilon})$ for any $\epsilon > 0$.
	\end{restatable}

	Theorems~\ref{thm:AlgOverview} and~\ref{thm:Ptodd} together cover almost all cases where $d_{\max} \leq 4$: $H=2P_2$ is hard by Theorem~\ref{thm:Ptodd}, $H = P_3+P_2$ is open, and all other cases with $d_{\max}\leq 4$ are linear-time solvable by Theorem~\ref{thm:AlgOverview}. Theorem~\ref{thm:AlgOverview} also gives linear-time algorithms for some cases where $d_{\max} > 4$; $H = P_3+2P_1$ and $H = 2P_2+P_1$ have $d_{\max} = 5$, and $H = P_2+3P_1$ has $d_{\max} = 6$. It appears that, algorithmically, the presence of a $P_1$ in $H$ helps out in structural analysis, which may explain the inability to attain a result for $H = P_3+P_2$. We further discuss particular cases and possible generalizations of our theorems in the conclusion.

	Our algorithmic results are mostly attained through careful analysis of the structure of the graph with respect to the forbidden pattern. This limits the ways in which a shortest path that realizes the diameter can appear in the graph. However, even for small patterns $H$, such analysis quickly becomes highly technical.

	We prove our hardness results in Section~\ref{sec:hardness} and give algorithmic results in Section~\ref{sec:algorithms}, and prove Theorems~\ref{thm:diamalggeneraloverview} and~\ref{thm:AlgOverview} in Section~\ref{sec:overviewproofs}. We discuss our results, conjectured generalizations, and open questions in the conclusion, see Section~\ref{sec:conclusion}.

	\section{Preliminaries}\label{sec:Prelims}
	
	A graph $G = (V,E)$ has a vertex set $V$ and edge set $E$, where we denote $|V| = n$ and $|E| = m$. Graphs are connected, undirected, and unweighted. For any $v\in V$, denote $N(v)$ as the neighbourhood of $v$, and $N[v] = N(v)\cup \{v\}$ as the neighbourhood including $v$. For a set of vertices $S\subseteq V$, let $G[S]$ denote the induced subgraph on the vertices of $S$, that is, the vertices of $S$ and all edges present in $G$ between those vertices. A vertex $v\in V$ is \emph{complete} to a set $S\subseteq V$ when $S \subseteq N(v)$, and \emph{anti-complete} to a set $S$ when $S\cap N(v) = \emptyset$. A set is $A\subseteq V$ is complete to a set $S\subseteq V$ when every vertex in $A$ is complete to $S$, and $A$ is anti-complete to $S$ when every vertex in $A$ is anti-complete to $S$. For vertices $v_1,\ldots,v_k\in V$, we use $\langle v_1,\ldots,v_k \rangle$ to denote an induced path from $v_1$ to $v_k$. $P_k$ denotes a path on $k$ vertices, for $k\geq 1$.
	
	A graph is $H$-free when it does not contain $H$ as an induced subgraph. For graphs $G,H$ let $G+H$ denote their disjoint union. A linear forest is a disjoint union of one or more paths. We also use a row of integers ${(a_i)}_{i=1}^k$ for some integer $k\geq 1$, to consider $H = \sum_{i=1}^{k}P_{a_i}$-free graphs, which models every possible linear forest for $H$. $H\subseteq_i G$ denotes that $H$ is an induced subgraph of $G$.
	
	Denote for two vertices $u,v$ the distance between $u$ and $v$ as $d(u,v)$. The diameter of a graph $G$ is the length of the longest shortest path, that is, $\diam(G) = \max_{u,v\in V} d(u,v)$. A pair of vertices $u,v\in V$ is called a diametral pair when $d(u,v) = \diam(G)$. A shortest path between a diametral pair $u,v$ of length $d(u,v)$ is called a diametral path. The problem of {\sc Diameter} is to decide on the value of $\diam(G)$ for a given graph $G$. In the literature, the variant is also considered where we have to report the diametral path. For our algorithms, this makes no difference, as we can always execute a Breadth First Search (BFS) to find diametral paths if we found a diametral pair. 
	
	We assume our graphs are connected, and so the diameter is never $\infty$. This is not a limiting assumption, as any BFS in time $O(n+m)$ can verify this, and our algorithms run in time $O(n+m)$. Also, for any graph, we can decide in time $O(n+m)$ whether it is a clique, and so whether its diameter is equal to 1. 
	
	For a graph class $\mathcal{G}$ (of connected graphs), denote $d_{\max}(\mathcal{G})$ as the supremum of the diameter over all graphs $G \in \mathcal{G}$, that is,
	\[d_{\max}(\mathcal{G}) = \sup_{G\in \mathcal{G}}\diam(G).\]
	We omit $\mathcal{G}$ when it is clear from context. For instance, for the class of cliques, $d_{\max} = 1$. For many graph classes, $d_{\max}$ is unbounded. However, we can consider $H$-free graphs where $H$ is some linear forest, as such cases exclude paths of a certain length as induced subgraphs, bounding the diameter. Let us prove this formally.
	\begin{theorem}\label{thm:HLongestPath}%
		Given a class $\mathcal{G}$ of connected $H$-free graphs, the following statements hold.
		\begin{enumerate}
			\item $d_{\max}(\mathcal{G})$ is bounded if and only if $H$ is a linear forest, and
			\item if $H$ is a linear forest, say $H = \sum_{i=1}^{k}P_{a_i}$ for some row of integers ${(a_i)}_{i=1}^k$ for some integer $k\geq 1$, then $d_{\max}(\mathcal{G}) = k - 3 + \sum_{i=1}^{k}a_i$.
		\end{enumerate}
	\end{theorem}
	\begin{proof}
		Assume for sake of contradiction that $\mathcal{G}$ is a class of connected $H$-free graphs where $H$ is not a linear forest. We may assume $H$ contains either a cycle or a vertex of degree at least 3. We can conclude that the graphs $G = P_t$ for every integer $t \geq 1$ are contained in $\mathcal{G}$. But then no bound on the diameter suffices, as we can always pick a $t\geq 1$ larger than the bound to find a graph with larger diameter. 

		Now assume $\mathcal{G}$ is a class of connected $H$-free graphs with $H$ a linear forest, say $H = \sum_{i=1}^{k}P_{a_i}$ for some row of integers ${(a_i)}_{i=1}^k$, for some integer $k\geq 1$. Note that any $G\in \mathcal{G}$ cannot contain an induced path on $k - 1 + \sum_{i=1}^{k}a_i$ vertices, as such a path would contain an induced copy of $H$. We see that the largest induced path in any graph $G\in \mathcal{G}$ has at most $k - 2 + \sum_{i=1}^{k}a_i$ vertices, and hence the length of any shortest path is at most $k - 3 + \sum_{i=1}^{k}a_i$. This provides the bound on the $\diam(G)$ for any graph $G\in \mathcal{G}$.
	\end{proof}
	
	Two vertices $u,v\in V$ are \emph{twins} when $N[u] = N[v]$, also called \emph{true twins}. Two vertices $u,v\in V$ are \emph{false twins} when $N(u) = N(v)$. We shall always be explicit when we talk about false twins; in general, twins will refer to true twins. A \emph{twin class} is a set of vertices all of which are pairwise twins. With respect to diametral paths, twins are not interesting, as they behave exactly the same in terms of distances. In some procedures, we would like to remove twins from the graph, such that only one each of the vertices in a twin class remains. Habib~et~al.~\cite{HabibPV98} showed that twins can be identified in time $O(n+m)$, using partition refinement techniques (see also~\cite{HermelinMLW19}).
	\begin{theorem}[\cite{HabibPV98}]\label{thm:twinremoval}
		Given a graph $G$, we can detect true and false twins in $G$ in $O(n+m)$ time.
	\end{theorem}
	
	Removing twins keeps most distance properties in a graph (see e.g.\ Coudert et al.~\cite{CoudertDP19}). We give the following theorem for completeness.

	\begin{proposition}\label{prop:twinremovaldists}
		Given a graph $G = (V,E)$, in $O(n+m)$ time we can detect true twins and remove all-but-one vertex from each twin class resulting in a graph $G' = (V',E')$ with $V'\subseteq V$, $E'\subseteq E$. The following hold:
		\begin{enumerate}
			\item[(i)] the distance between two vertices $u,v$ in $G'$ is equal to the distance between $u,v$ in $G$,
			\item[(ii)] the diameter of $G'$ is equal to the diameter of $G$, unless $G$ is a clique,
			\item[(iii)] if $G$ is $H$-free for some graph $H$ then $G'$ is also $H$-free. 
		\end{enumerate}
	\end{proposition}
	\begin{proof}
		Detect classes of true twins in linear time using Theorem~\ref{thm:twinremoval}.
		We can remove all-but-one vertex of each twin class, by first marking vertices for removal, and then in linear time enumerating the graph to create a new graph without the marked vertices and their edges. This is the graph $G'$. Let $\phi : V \rightarrow V'$ map vertices of $G$ to its twin in $G'$ that was not removed; note that this is the identity function for vertices with no twins in $G$.

		If $G'$ has only one vertex, then there are no shortest paths. This can be the case when $G$ is a clique.

		Let $u',v'\in V'$. Let $u,v\in V$ be the same vertices in $G$, and denote the shortest path between $u$ and $v$ with $\langle u = w_0, w_1, \ldots, w_k = v \rangle$. Now $\langle u' = \phi(w_0), \phi(w_1), \ldots, \phi(w_k) = v'\rangle$ must be a shortest path from $u'$ to $v'$ in $G'$, as twins have the exact same neighbourhood, and distances cannot be shortened by vertex removal.

		Assume $G$ is $H$-free for some graph $H$, and assume for sake of contradiction that $G'$ is not $H$-free. Then, there is some set of vertices $A'\subseteq V'$ such that $G'[A']$ forms an induced copy of $H$. Take for each $v' \in A'$ some vertex in $v\in V$ such that $\phi(v) = v'$, and call the set of these vertices $A$. Because $v'$ and $v$ are twins, it must be that $G[A]$ is also an induced copy of $H$ in $G$, a contradiction.
	\end{proof}

	This procedure does not only work for true twins. Assume we have a set $B \subset V(G)$ and a set $V(G)\setminus B$ where we would like to partition $B$ into classes of false twins with respect to the edges towards $V(G)\setminus B$ (so irrespective of edges in $G[B]$). We note that we can do this in linear time, by adjusting the algorithm of Theorem~\ref{thm:twinremoval} by having the base set be $B$ and refining (partitioning the base set) only for every $x\in V(G)\setminus B$ on $N(x) \cap B$. We get the following corollary, which was proved by Ducoffe~\cite{ducoffe2022optimal} (formulated in terms of modules):
	\begin{corollary}\label{cor:bipartitetwins}
		Given a graph $G$ and a vertex set $B\subseteq V(G)$, we can partition $B$ into classes of false twins with respect to their neighbourhoods towards $V(G)\setminus B$ in time $O(n+m)$.
	\end{corollary}
	
	\section{Proofs of Theorem~\ref{thm:hardnessgeneral} and~\ref{thm:Ptodd}}\label{sec:hardness}
	
	In general, showing hardness for {\sc Diameter} for some $H$ also shows hardness for all $H'$-free graphs where $H'$ with $H \subseteq_i H'$ is a graph that contains $H$ as an induced subgraph. Showing a linear-time, i.e. $O(n + m)$ time, algorithm for {\sc Diameter} for some $H$ also shows linear-time algorithms for all $H'$-free graphs, where $H' \subseteq_i H$ is a graph contained as an induced subgraph in $H$.

	\HardnessGeneral*
	\begin{proof}
		It is well-known that deciding between diameter 2 or 3 on split graphs is hard~\cite{RodittyWilliams13}; see also Appendix~\ref{appendix} for an illustration and proof of this result. Split graphs are $(2P_2, C_4, C_5)$-free~\cite{hammer1977split}.
		Hence, if $H$ contains a cycle of length at least 4 as an induced subgraph, then computing diameter on $H$-free graphs is hard by the split graph construction (it is chordal). If $H$ contains a triangle, then computing the diameter is hard by a hardness construction on bipartite graphs~\cite{AbboudWW16}. So $H$ does not contain a cycle.
		If $H$ contains a vertex of degree at least three, computing the diameter on $H$-free graphs is hard, as the construction by Evald and Dahlgaard~\cite{EvaldDahlgaard16} can be made claw-free by adding chords to all binary trees in the construction (without affecting correctness).
		Hence, the remaining cases are where $H$ is a linear forest.
		As the split graph construction shows hardness for $2P_2$-free graphs, deciding between diameter 2 and 3 is hard for any $H$-free graph where the linear forest $H$ contains a $2P_2$.
	\end{proof}
	
	The split graph hardness construction by Roditty and Williams~\cite{RodittyWilliams13} shows that deciding between diameter 2 or 3 is hard for all $P_t$-free graphs with $t~\geq~5$. However, many such $P_t$-free graphs may have diameter much larger than just 2 or 3; in particular, $d_{\max}$ may be much larger. Therefore, this hardness construction does not necessarily rule out being able to compute $d_{\max}$ for many such $H = P_t$. To mitigate this, we extend the hardness construction by Roditty and Williams~\cite{RodittyWilliams13} to rule out solving \dmaxproblem{} in linear time for all $H = P_t$-free graphs for odd $t$, $t\geq 5$.
	
	\OddPtHard*
	\begin{proof}
		In this proof we heavily use the split graph hardness construction by Roditty and Williams~\cite{RodittyWilliams13}; see also Appendix~\ref{appendix} for the construction by Roditty and Williams. The construction shows that we cannot decide between diameter 2 or 3 on $m$-edge (sparse) split graphs in time $O(m^{2-\epsilon})$ for all $\epsilon > 0$ unless SETH fails.
		Split graphs are $2P_2$-free~\cite{hammer1977split}. For $2P_2$-free graphs $d_{\max}$ = 3 by Theorem~\ref{thm:HLongestPath}, so hardness for this class immediately follows.

		Let $t\geq 5$ be an odd integer. We prove that deciding on $d_{\max}({P_t\text{-free}})$ is hard under SETH, that is, we show hardness for deciding between diameter $d_{\max} = t-2$ and $d_{\max} - 1 = t-3$, which are the correct values by Theorem~\ref{thm:HLongestPath}.
		
		Let $G$ be any split graph hardness instance with $n$ vertices and $m$ edges. As $G$ is a split graph, it is $2P_2$-free and so $P_5$-free. For $t=5$ this proves the theorem. Consider the following augmentation: add to every vertex in the independent set a path of $c = \frac{t-5}2$ edges. This adds at most $cn$ vertices and edges to the graph. Let $G'$ denote the resulting graph with $O(n + cn) = O(n)$ vertices and $O(m + cn) = O(m)$ edges. The diameter of $G'$ is either $2 + 2c$ or $3 + 2c$, depending on the diameter of $G$. Deciding on the diameter of $G'$ directly implies deciding on the diameter of $G$, so the lower bound carries over. Therefore, the hardness of deciding on the diameter of $G'$ holds for deciding between $3 + 2c = t - 2$ and $2 + 2c = t-3$. Because induced paths have been extended by at most $2c$ vertices and $G$ is $P_5$-free, we get that $G'$ is $P_{5 + 2c}$-free. The theorem follows. 
	\end{proof}
	
	For even $t$, the construction does not work, as the maximum diameter of the constructed graphs is always odd. We conjecture that Theorem~\ref{thm:Ptodd} generalizes to all $t\geq 5$, not just odd $t$; see the conclusion for a discussion. Note that when $t \leq 4$, the diameter is at most 2, and so it suffices to check whether the graph is a clique, so indeed hardness can only hold for $t\geq 5$.
	
	\section{Algorithmic Results}\label{sec:algorithms}

	\subsection{\texorpdfstring{$(P_2+2P_1)$}{(P2+2P1)}-free graphs}
	
	In this section we prove that the diameter of a $(P_2 + 2P_1)$-free graph can be computed in linear time. The statement of the theorem is slightly stronger however, as we will need this algorithm as a subroutine in another algorithm later on.
	
	\begin{figure}
		\centering
		\includegraphics[width=0.6\textwidth]{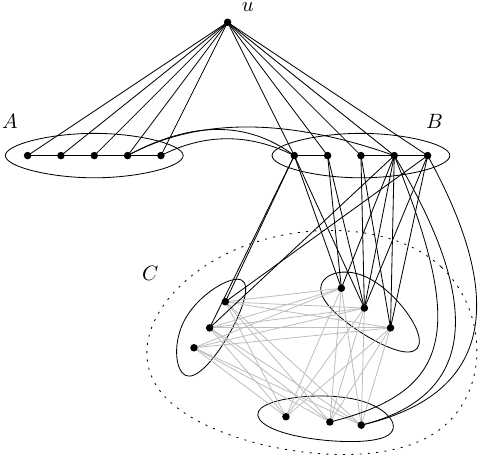}
		\caption{A sketch of a $(P_2 + 2P_1)$-free graph as seen from some vertex $u$.}\label{fig:p2+2p1-freediameter}
	\end{figure}
	\begin{theorem}\label{thm:2P1+P2free}
		Given a graph $G$, there is an algorithm that in $O(n+m)$ time either (a) correctly decides that $G$ is not $(P_2+2P_1)$-free; or (b) outputs a shortest path, which is diametral if $G$ is $(P_2+2P_1)$-free. 
	\end{theorem}
	\begin{proof}
		Let $G = (V,E)$ be a graph.
		The diameter of any $(P_2+2P_1)$-free graph is at most 4 by Theorem~\ref{thm:HLongestPath}. If the diameter of $G$ is 1, then the graph is a clique, which we can check in $O(n + m)$ time, and return any arbitrary pair of vertices, or a single vertex if $|V| = 1$.
		
		Remove twins from the graph in $O(n+m)$ time. By Proposition~\ref{prop:twinremovaldists}, distances and the diameter are not affected, and the graph is $H$-free if it was $H$-free for some graph $H$. By abuse of notation, we still call this graph $G = (V,E)$.

		Let $u$ be a vertex in $G$ with lowest degree, which can be found in $O(n+m)$ time. Now execute a BFS from $u$. We distinguish the structure of the graph as seen from $u$; see Figure~\ref{fig:p2+2p1-freediameter}. Let $C = V \setminus N[u]$ and let $A \subseteq N(u)$ be the subset of vertices of $N(u)$ with no edges to $C$. Let $B = N(u) \setminus A$. Note that $A,B,C$ can be identified by the BFS from $u$. If $C = \emptyset$, then the diameter of $G$ is at most 2, and we are done. Hence, $C\neq \emptyset$ and $B\neq \emptyset$. We note that $A = \emptyset$, which can be seen by the following. For any $a\in A$ it holds that $\deg(a) \leq \deg(u)$ by definition of $A$. $u$ was picked to be a vertex of lowest degree in $G$, so for any $a\in A$ we have $\deg(a) = \deg(u)$. But as $N(a) \subseteq \{u\} \cup A \cup B$ for all $a\in A$ and $N(u) = A \cup B$, it follows that every $a\in A$ is a twin of $u$. But then $A = \emptyset$ as we removed twins from $G$.

		If there is a vertex at distance 5 or more from $u$, then return that the graph is not $(P_2+2P_1)$-free. If there is a vertex at distance 4 from $u$, then return this shortest path; it is diametral if $G$ is $(P_2+2P_1)$-free. Both of these cases are identified by the BFS from $u$. Now observe that any shortest path of length~3~or~4 must have at least one endpoint in $C$, as the distances between vertices in $A\cup B \cup \{u\}$ are at most 2 by $u$.
		
		We prove properties of $G$ under the assumption that it is $(P_2+2P_1)$-free.

		\begin{claim}\label{clm:P2+2P1CompleteRandNNGC}
			If $G$ is $(P_2+2P_1)$-free, then (a) $G[C]$ is a complete $r$-partite graph for some $r\geq 1$; and (b) every $b\in B$ has at most one non-neighbour in every part of $G[C]$.
		\end{claim}
		\begin{proof}
			To prove (a), because $u$ is non-adjacent to all of $C$, $G[C]$ is a $(P_2 + P_1)$-free graph. The complement of a $(P_2 + P_1)$-free graph is $P_3$-free, which is a disjoint union of cliques, and the complement of a disjoint union of cliques is a complete $r$-partite graph, for some $r\geq 1$.

			To prove (b), for sake of contradiction, assume that a vertex $b\in B$ has two non-neighbours $c_1,c_2$ in a single part of $G[C]$. Then, $(u,b)$ together with $c_1,c_2$ form a $P_2 + 2P_1$, a contradiction.
		\end{proof}

		\begin{claim}\label{clm:P2+2P1DetectRpartite}
			In $O(n+m)$ time, we can decide whether $G[C]$ is a complete $r$-partite graph for some $r \geq 1$ and, if so, return its parts.
		\end{claim}
		\begin{proof}
			Detect false twins in $G[C]$ in $O(n+m)$ time using Theorem~\ref{thm:twinremoval}. Let $r$ be the number of false twin classes, note that $r\leq n$. Initialize an array of size $r$ in $O(n)$ time, and count for each class the number of vertices in it in $O(n)$ time total. Now we have to check that each vertex is complete to all vertices except those in its class. To do this, iterate for each vertex $v$ over its adjacency list and count the number of neighbours in $C$. The degree of a vertex $v$ should come out to be the number of vertices in $C$ minus the size of its class, which are both known.

			If this process succeeds, we have $r$ classes, let us call them parts. Every vertex in a part is complete to all other parts, and has no neighbours in its part. Hence, this is a complete $r$-partite graph. We can return the parts by returning each false twin class.
		\end{proof}

		\begin{claim}\label{clm:P2+2P1ClaimDiam4}
			In $O(n+m)$ time, either (a) we find a length-$4$ shortest path with both endpoints in $C$; (b) we find a length-$3$ shortest path with both endpoints in $C$ and conclude no such length-$4$ shortest path exists in $G$; (c) we conclude no length-$3$ or length-$4$ shortest path with both endpoints in $C$ exists in $G$; or (d) we conclude $G$ is not $(P_2+2P_1)$-free.
		\end{claim}
		\begin{proof}
			First execute the algorithm of Claim~\ref{clm:P2+2P1DetectRpartite} to detect if $G[C]$ is a complete $r$-partite graph for some integer $r\geq 1$. If $G[C]$ is not a complete $r$-partite graph for some $r\geq 1$, by Claim~\ref{clm:P2+2P1CompleteRandNNGC}, we can safely return option (d).
			If $G[C]$ is complete $r$-partite for some $r > 1$, then the distances between vertices of $C$ are at most 2, and a shortest path of length 3 or 4 with both endpoints in $C$ does not exist, and we return option (c).

			If $G[C]$ is complete $1$-partite, that is, $G[C]$ is an independent set, then look at the neighbourhood of any arbitrary $b\in B$. By Claim~\ref{clm:P2+2P1CompleteRandNNGC}, if $G$ is $(P_2+2P_1)$-free, $b$ is non-adjacent to at most one vertex $c\in C$. If $b$ is adjacent to all of $C$, then return no shortest path of length 3 or 4 with both endpoints in $C$ exist in $G$, as all vertices in $C$ have $b$ as a common neighbour, so return option (c). If $b$ has more than one non-adjacency in $C$, then return option (d). Otherwise, $c$ is the only candidate for an endpoint of a length-3 or length-4 shortest path, as all other pairs $c_1,c_2 \in C$ have $b$ as a common neighbour. Execute a BFS from $c$ and return a length-4 shortest path if found, option (a). If a longer shortest path is found, then return that $G$ is not $(P_2+2P_1)$-free, option (d). If instead only a length-3 shortest path is found, return it as option (b), and by this analysis, there is no shortest path of length 4 with both endpoints in $C$.
		\end{proof}

		Run the algorithm of Claim~\ref{clm:P2+2P1ClaimDiam4}. If it returns option (d), then output that $G$ is not $(P_2+2P_1)$-free. If it returns option (a), then output the length-4 shortest path the algorithm gives; it is diametral if $G$ is $(P_2+2P_1)$-free. In both other cases, we argue no length-4 diametral path can exist if $G$ is $(P_2+2P_1)$-free. If $G$ is $(P_2+2P_1)$-free, distances from vertices in $B$ to vertices in $C$ are at most 3, because $G[C]$ is $r$-partite and every $b\in B$ has at most one non-adjacent vertex per part of $G[C]$ by Claim~\ref{clm:P2+2P1CompleteRandNNGC}. We already found a length-4 shortest path with $u$ as an endpoint, if it exists. But then any length-4 shortest path has both endpoints in $C$, if it exists, as we already knew at least one endpoint was in $C$. In both options (b) and (c) we can conclude that no length-4 shortest path with both endpoints in $C$ exists. 

		We continue as follows.
		If $u$ has a vertex at distance $3$, then we would already know this by the BFS from $u$. Otherwise, every $c\in C$ is at distance~$2$ from $u$. If the algorithm of Claim~\ref{clm:P2+2P1ClaimDiam4} returned option (b), we can output a shortest path of length $3$ with both endpoints in $C$; it is diametral if $G$ is $(P_2+2P_1)$-free.
		If the algorithm of Claim~\ref{clm:P2+2P1ClaimDiam4} returned option (c), no length-$3$ shortest path with both endpoints in $C$ exists in $G$. Hence, the only remaining case for a length-$3$ shortest path is that there is a vertex in $B$ with distance $3$ to a vertex in $C$. We prove a claim on the structure of $G[C]$ if $G$ is $(P_2+2P_1)$-free and a shortest path of length $3$ exists from a vertex in $B$ to a vertex in $C$.

		\begin{claim}\label{clm:P2+2P1ClaimLength3Struc}
			If $G$ is $(P_2+2P_1)$-free and a length-$3$ shortest path exists from some $b\in B$ to some $c\in C$, then (a) $G[C]$ has exactly one part with more than one vertex, and (b) all vertices in $B$ with a vertex in $C$ at distance $3$ are adjacent to only that multi-vertex part and have exactly one non-neighbour in that part.
		\end{claim}
		\begin{proof}
			Assume $G$ is $(P_2+2P_1)$-free and a length-$3$ shortest path exists from some vertex in $B$ to some vertex in $C$. By Claim~\ref{clm:P2+2P1CompleteRandNNGC}, $G[C]$ is a complete $r$-partite graph for some $r\geq 1$.
			First, note that any $b\in B$ adjacent to multiple parts has distance at most 2 to all vertices in $C$, because $G[C]$ is complete $r$-partite. So, any $b\in B$ with a vertex at distance $3$ in $C$ is adjacent to exactly one part of $G[C]$. By assumption, there is such a vertex in $B$. Moreover, by Claim~\ref{clm:P2+2P1CompleteRandNNGC}, every $b\in B$ has at most one non-neighbour in every part of $G[C]$. Hence, there is at most one part with multiple vertices in $G[C]$. If there is no part with multiple vertices, then every $b\in B$ has distance at most~2 to all $c\in C$. This proves (a).
			To prove (b), again note that any $b\in B$ adjacent to multiple parts has distance at most 2 to all vertices in $C$. By Claim~\ref{clm:P2+2P1CompleteRandNNGC}, any $b\in B$ is adjacent to the multi-vertex part and has at most one non-neighbour in that part. Having zero non-neighbours in the multi-vertex part implies the distance to all vertices in $C$ is at most $2$. The claim follows.
		\end{proof}
		
		By Claim~\ref{clm:P2+2P1ClaimLength3Struc}, we can look for $G[C]$ to have simple structure. In particular, only one part may have multiple vertices.
		By Claim~\ref{clm:P2+2P1DetectRpartite}, we can detect if $G[C]$ is complete $r$-partite in $O(n+m)$ time, and, given the parts, check whether only one part has multiple vertices in $O(n+m)$ time.
		If this is not the case, then there is no length-$3$ shortest path from a vertex in $B$ to a vertex in $C$ by Claim~\ref{clm:P2+2P1ClaimLength3Struc}, if $G$ is $(P_2+2P_1)$-free, and we may return a length-$2$ shortest path with as witness some non-adjacent pair of vertices and a common neighbour.

		Otherwise, the structure is as Claim~\ref{clm:P2+2P1ClaimLength3Struc}(a) and (b) suggest. Find all $b\in B$ only adjacent to the multi-vertex part with one non-neighbour in that part in $O(n + m)$ time, by iterating over the adjacency lists of the vertices in $B$. Let this be the set of vertices $B'$. Then look for each vertex in $B'$ whether all its neighbours in $B$ have the same non-adjacency in the multi-vertex part in $O(n + m)$ time. If there is a vertex $b\in B'$ that meets this requirement, and $G$ is $(P_2+2P_1)$-free, then this is a witness for diameter-3 shortest path: $b$ is non-adjacent to one vertex $c\in C$ in a multi-vertex part of $G[C]$, which is the only part it is adjacent to, and $N(b) \cap N(c) = \emptyset$. So the distance from $b$ to $c$ is at least (and at most) 3. To verify, execute a BFS from any single one of these vertices. If a shortest path is found of length 4 or longer, then return that $G$ is not $(P_2+2P_1)$-free. If a length-3 shortest path is found, then return it. Otherwise, there is no length-3 shortest path in $G$ from $B$ to $C$, if $G$ is $(P_2+2P_1)$-free.
		
		If in none of the above cases a confirmation for a shortest path of length 3 or 4 was found, and the graph is $(P_2+2P_1)$-free and not a clique, then the diameter of $G$ must be 2. Return some non-adjacent pair with some common neighbour as a shortest path in $O(n+m)$ time.
	\end{proof}

	As an immediate corollary, we get that we can compute the diameter of $3P_1$-free graphs in linear time, as any $3P_1$-free graph is $(P_2 + 2P_1)$-free.
	
	\begin{corollary}\label{cor:3P1free}
		Given a $3P_1$-free graph $G$, we can compute the diameter of $G$ and return a diametral path of $G$ in $O(n+m)$ time.
	\end{corollary}

	\subsection{\texorpdfstring{$(P_3+P_1)$}{(P3+P1)}-free graphs}

	We use a characterization by Olariu~\cite{olariu1988paw} to argue we can find the diameter of $(P_3+P_1)$-free graphs in linear time. As a subroutine, we call on the algorithm of Theorem~\ref{thm:2P1+P2free}.

	\begin{theorem}\label{thm:P3+P1free}
		Given a $(P_3 + P_1)$-free graph $G$, we can compute the diameter of $G$ and return a diametral path of $G$ in $O(n+m)$ time.
	\end{theorem}
	\begin{proof}
		First note that the complement of a $P_3+P_1$ is a paw. By a result by Olariu~\cite{olariu1988paw}, a graph is paw-free if and only if each component is triangle-free or complete multi-partite. This characterization must hold for the complement of our $(P_3 + P_1)$-free input graph $G$. Hence, our input graph $G$ can be divided into parts, where each part is complete to every other part, and each part is either $3P_1$-free or a $P_3$-free graph (the complements of triangle-free graphs and complete multi-partite graphs).

		If $G$ consists of only one part, then it is either a clique with diameter at most 1, or it is $3P_1$-free. When $G$ consists of multiple parts, then its diameter can only be 2 if it is not a clique. The algorithm is now given by the following. Check in $O(n+m)$ time whether $G$ is a clique. If it is a clique, return some arbitrary pair of vertices or a single vertex if $|V| = 1$. If $G$ is not a clique, then run the algorithm of Theorem~\ref{thm:2P1+P2free}. If $G$ consists of only one part, then the algorithm will return a shortest path that is diametral, because any $3P_1$-free graph is also $(P_2+2P_1)$-free. If $G$ instead consists of multiple parts, then there are no shortest paths of length more than~2 in $G$, so the algorithm will either return that $G$ is not $(P_2+2P_1)$-free or return a length-2 shortest path. In either case, we can correctly conclude that the diameter of $G$ is 2, and return a path consisting of two arbitrary non-adjacent vertices and a common neighbour.
	\end{proof}

	\subsection{\texorpdfstring{$(P_4+P_1)$}{(P4+P1)}-free graphs}
	
	We next show that we can decide whether the diameter of a $(P_4 + P_1)$-free graph is equal to $d_{\max}$ in $O(n+m)$ time. The proof will identify all possibilities of a diameter-4 path occurring in relation to a BFS from an arbitrary vertex. Luckily, most cases reduce to some other case in the proof, and algorithmically speaking, only a few cases require algorithmic computation.
	
	\begin{theorem}\label{thm:P4+P1free}
		Given a $(P_4 + P_1)$-free graph $G$, we can decide whether the diameter of $G$ is equal to $d_{\max} = 4$ in $O(n+m)$ time.
	\end{theorem}
	\begin{proof}
		Let $G = (V,E)$ be a (connected) $(P_4 + P_1)$-free graph. Indeed, $d_{\max} = 4$ by Theorem~\ref{thm:HLongestPath}.
		We view the structure of $G$ from a BFS from an arbitrary vertex $u$. Let $C = V \setminus N[u]$ and denote $B = N(u)$. If $C = \emptyset$ or $B = \emptyset$, then the diameter of $G$ is at most 2, so assume this is not the case. Note that both sets can be identified during the BFS from $u$ with no overhead. Moreover, $G[C]$ is $P_4$-free. We use the convention that $b_i \in B$ and $c_i \in C$.
		
		\begin{figure}[tb]
			\centering
			\includegraphics[width=.9\textwidth]{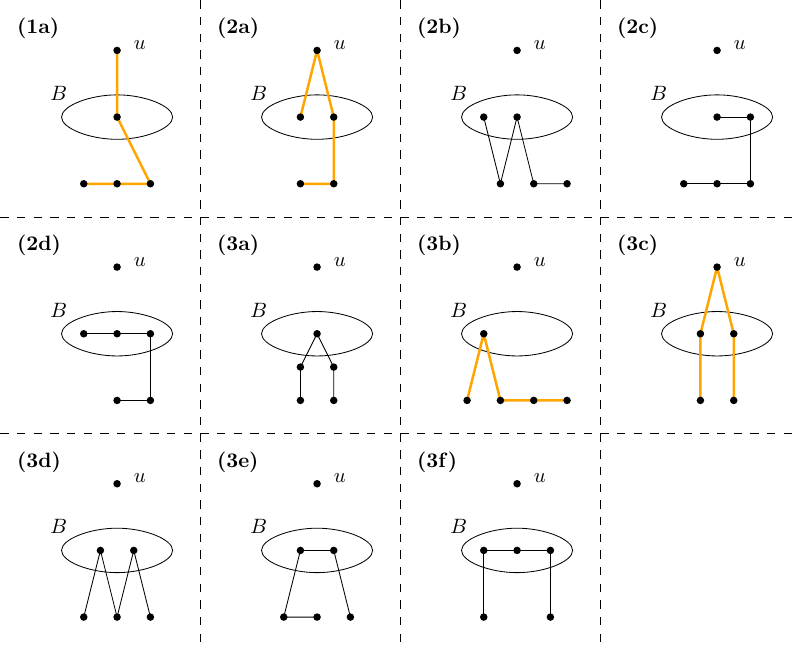}
			\caption{An illustration of the types of a diameter-4 shortest path appearing in a $(P_4+P_1)$-free graph, with respect to the vertex $u$ and its neighbourhood $N(u) = B$. Only the highlighted types \textbf{(1a)}, \textbf{(2a)}, \textbf{(3b)}, \textbf{(3c)} require algorithmic computation.}\label{fig:P4+P1cases}
		\end{figure}
		
		We first list all possibilities for a diameter-4 shortest path to exist with respect to its structure; see Figure~\ref{fig:P4+P1cases} for an illustration. We will call each such possibility a `type'. First, observe that at least one of the ends of any distance-4 shortest path must be in $C$, because distances in $G[B\cup \{u\}]$ are at most 2 via $u$. Here, we do not write down symmetries, because in an undirected graph the path $\langle v_1,\ldots,v_k \rangle$ is equivalent to the path $\langle v_k,\ldots, v_1\rangle$. We now distinguish the following cases:\\
		\textbf{(1)} If the path has $u$ as one endpoint and the other endpoint is in $C$,	then the only possibility is type \textbf{(1a)}~$\langle u,b,c_1,c_2,c_3 \rangle$, as any vertex in $B$ has $u$ as a neighbour.\\
		\textbf{(2)} If the path has one end in $B$ and the other in $C$, then we have several options, depending on the number of vertices of $B$ on the path. We have at least one vertex of $B$ on the path because one endpoint is in $B$. We cannot have only one vertex in $B$ on the path, because then the rest of the path is in $C$ but $G[C]$ is $P_4$-free. If we have exactly two vertices of $B$ on the path, they can be adjacent or non-adjacent. Two non-adjacent vertices of $B$ on the path give the possibilities of type \textbf{(2a)}~$\langle b_2,u,b_1,c_1,c_2 \rangle$ and type \textbf{(2b)}~$\langle b_1,c_1,b_2,c_2,c_3 \rangle$; the second vertex of $B$ cannot appear later on the path because there would be a shorter path via $u$. Two adjacent vertices of $B$ on the path give type \textbf{(2c)}~$\langle b_1,b_2,c_1,c_2,c_3 \rangle$, and no other option exists. If we have three vertices in $B$ on the path the only option is type \textbf{(2d)}~$\langle b_1,b_2,b_3,c_1,c_2 \rangle$, because if some vertex of $B$ appears later on the path there would be a shorter path via $u$. There can be at most three vertices in $B$ on the path, because distances between vertices of $B$ are at most 2 via $u$. Hence, these are all the options for a path with one endpoint in $B$ and one in $C$.\\
		\textbf{(3)} The last option is for the path to have both endpoints in $C$. Again, there are multiple cases depending on the number of vertices of $B$ on the path. There cannot be zero vertices of $B$ on the path, because $G[C]$ is $P_4$-free. If there is only one vertex of $B$ on the path, we have type \textbf{(3a)}~$\langle c_1,c_2,b,c_3,c_4 \rangle$ and type \textbf{(3b)}~$\langle c_1,b,c_2,c_3,c_4 \rangle$, and we cannot pass by $u$. If we have exactly two vertices of $B$ on the path and they are non-adjacent, we have type \textbf{(3c)}~$\langle c_1,b_1,u,b_2,c_2 \rangle$ and type \textbf{(3d)}~$\langle c_1,b_1,c_2,b_2,c_3 \rangle$. If we have two vertices of $B$ on the path and they are adjacent, we have type \textbf{(3e)}~$\langle c_1,c_2,b_1,b_2,c_3 \rangle$. If we have exactly three vertices of $B$ on the path we have type \textbf{(3f)}~$\langle c_1,b_1,b_2,b_3,c_2 \rangle$, which is clearly the only option. There can be at most three vertices in $B$ on the path, because distances between vertices of $B$ are at most 2 via $u$. Hence, these are all the options for a path with both endpoints in $C$.\\
		Because we listed all length-4 shortest paths with one endpoint in $C$ and the other endpoint in $u$, $B$, or $C$, and $V = \{u\} \cup B \cup C$, these must be all the options for a length-4 shortest path to appear in $G$.
		
		To prove the computability of some types, the following structural observation will come in useful.
		\begin{claim}\label{clm:completetoallC}
			If a vertex $b\in B$ has a neighbour and a non-neighbour in a single component of $G[C]$, then $b$ is complete to all other components of $G[C]$.
		\end{claim}
		\begin{proof}
			Let $C_1\subseteq C$ denote the vertices of a component of $G[C]$ such that $b$ has a neighbour and a non-neighbour in $C_1$. Let $c_1'$ denote a neighbour of $b$ and $c_2'$ a non-neighbour. There must be a path from $c_1'$ to $c_2'$ within $G[C_1]$, because $G[C_1]$ is a connected component. But then there must be two vertices $c_1,c_2\in C_1$ with $(c_1,c_2) \in E$ on this path such that $c_1$ is a neighbour of $b$ and $c_2$ is a non-neighbour of $b$.
			Notice that $\langle c_1, c_2, b, u\rangle$ forms an induced $P_4$. Consider a vertex $c_3 \in C$ in another component of $G[C]$. It has no edges to either $c_1$ or $c_2$, as it is in another component, and it is non-adjacent to $u$, as it is in $C$. But then, as the graph is $(P_4+P_1)$-free, $(c_3,b)$ must be an edge of the graph.
		\end{proof}
		
		We show we can decide whether the diameter of $G$ is 4 in time $O(n+m)$ by resolving each type. Algorithmically speaking, only types \textbf{(1a)}, \textbf{(2a)}, \textbf{(3b)}, and \textbf{(3c)} will require computation to find diametral paths corresponding to them. We will show that all other types are either covered by these computations, or are non-existent in $G$.
		\paragraph{(1a)~$\langle u,b,c_1,c_2,c_3 \rangle$} This type is identified by the initial BFS from $u$ if and only if it occurs in $G$.
		\paragraph{(2a)~$\langle b_2,u,b_1,c_1,c_2 \rangle$} For this type, let us further partition $C$ into sets $C'$ and $D$, where $D$ consists of the vertices with no neighbours in $B$ (i.e.\ the vertices at distance 3 from $u$), and $C' = C \setminus D$. It must be that $c_2 \in D$ for a diametral path of type \textbf{(2a)} to exist, otherwise, $\langle c_2, b_i, u, b_2 \rangle$ is a shorter path for some $b_i\in B$. So rename the vertex $c_2$ as $d = c_2 \in D$ and look for a path $\langle b_2,u,b_1,c_1,d \rangle$. Further partition $C'$ into $C_1$ and $C_2$, where $C_1$ are the vertices with edges towards $D$ and $C_2 = C' \setminus C_1$. This partitioning can be done during the BFS from $u$, or alternatively using another linear pass over all vertices and edges. Note that every vertex in $D$ has at least one neighbour in $C_1$; otherwise, a diametral path of type \textbf{(1a)} exists in $G$, and we are done.
		Let us first describe the algorithmic steps necessary for this type.
		\paragraph{(2a.algorithm)}
		Find a vertex $d\in D$ with the smallest degree with respect to $C_1$, and execute a BFS from $d$. If a distance-4 vertex is found, return that the diameter of $G$ is 4. If we do not find a vertex at distance 4, no distance-4 diametral path of type \textbf{(2a)} exists in $G$.
		\medskip
		
		We next prove correctness of \textbf{(2a.algorithm)}. We prove correctness when $G[D]$ is not connected in \textbf{(2a.1)} and correctness when $G[D]$ is connected in \textbf{(2a.2)}. To do this, we analyse the structure of $G$ under the assumption that a diametral path of type \textbf{(2a)} exists, to conclude the structure of $G$ must then be `simple' in some way, to the extent that the above algorithmic steps suffice.
		
		\paragraph{(2a.1)} Assume $G[D]$ is not connected. We first prove that every vertex in $C_1$ is complete to $D$. Assume for sake of contradiction that there is a $c\in C_1$ which is not complete to $D$. Let $d' \in D$ be a non-neighbour of $c$. Let $d\in D$ be some neighbour of $c$, which exists because $c\in C_1$. Now $\langle d,c,b,u\rangle$ is an induced $P_4$ for some $b\in B$ which exists by definition of $C_1$. But then $d'$ must be a neighbour of $d$; otherwise, it would induce a $P_4+P_1$. We see that every neighbour of $c$ is adjacent to every non-neighbour of $c$ in $D$. We get a contradiction with the assumption that $G[D]$ is not connected.
		
		So, every vertex in $C_1$ must be complete to $D$. But then, from the viewpoint of shortest paths, for any $b\in B$, the distances from all $d\in D$ to $b$ must be equal, as the shortest path to $b$ must go through some vertex of $C_1$ in the first step, and $C_1$ is complete to $D$. Hence, if the BFS that \textbf{(2a.algorithm)} executes does not find a diametral path of length 4, no diametral path of type \textbf{(2a)} exists in $G$.
		
		\paragraph{(2a.2)} Assume $G[D]$ is connected. Then $G[C_1 \cup D]$ is a connected cograph, so it has diameter at most 2. Let $B_1 \subseteq B$ be the vertices of $B$ with neighbours in $C_1$. Vertices in $B_2 = B \setminus B_1$ have no neighbours in $C_1$. Every vertex in $B_1$ has distance at most 3 to any $d\in D$, as the diameter of $G[C_1\cup D]$ is at most 2. So, for a shortest path of type \textbf{(2a)} $\langle b_2,u,b_1,c_1,d \rangle$ we get $b_2\in B_2$ and $b_1\in B_1$.
		
		Let us call a pair $(d,b_2)$ with $d\in D$, $b_2\in B_2$ `good' if $d$ has distance 4 to $b_2$, and `bad' when $d$ has distance at most $3$ to $b_2$.
		Assuming a diametral path of form $\langle b_2,u,b_1,c_1,d \rangle$ exists in $G$, with $b_2\in B_2$, $b_1 \in B_1$, $c_1\in C_1$, $d\in D$, it is clear that $(d,b_2)$ is good.
		Assume that we also have that $(d',b_2)$ is bad for some $d \neq d' \in D$. Then $d'$ has distance exactly 3 to $b_2$, as $b_2$ has no neighbour in $C_1$. Let $c_1'$ be the neighbour of $d'$ on some distance-3 path from $d'$ to $b_2$. Then $c_1 \neq c_1'$; otherwise, $d$ has distance 3 to $b_2$. Then the shortest path from $d'$ to $b_2$ is of the form $\langle d', c_1', c_2', b_2 \rangle$ with $c_2'\in C_2$ (case \textbf{(2a.2.1)}), or $\langle d', c_1', b_1', b_2 \rangle$ with $b_1' \in B_1$, $b_1'\neq b_1$ (case \textbf{(2a.2.2)}). See Figure~\ref{fig:Case2aStructure} for an illustration of both scenarios.
		
		\begin{figure}[t]
			\centering
			\includegraphics{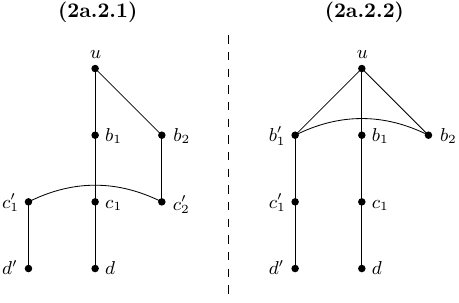}
			\caption{Structure for type \textbf{(2a)} with respect to a $d\in D$ and a $d' \in D$ for which $(d,b_2)$ is good and $(d',b_2)$ is bad for $b_2 \in B_2$. The path from $d'$ to $b_2$ goes through either some $c_2'\in C_2$ (left, \textbf{(2a.2.1)}) or some $b_1'\in B_1$ (right, \textbf{(2a.2.2)}).
			}\label{fig:Case2aStructure}
		\end{figure}
		
		\paragraph{(2a.2.1)} The shortest path from $d'$ to $b_2$ is of the form $\langle d', c_1', c_2', b_2 \rangle$ with $c_2'\in C_2$.
		Note that $(c_1,u),(c_1',u)\notin E$ by definition of $C$, and $(c_1,b_2),(c_1',b_2)\notin E$ by definition of $B_2$, and $(b_1,b_2), (d,c_1'), (c_1,c_2'), (d,c_2')\notin E$ as $(d,b_2)$ is good. But then $(c_1',c_1) \in E$ or $(c_1', b_1)\in E$; otherwise, $\langle c_1, b_1, u, b_2 \rangle + c_1'$ is an induced $P_4 + P_1$ in the graph.
			\begin{itemize}
				\item If $(c_1',c_1) \in E$, then $\langle c_2', c_1', c_1, d \rangle + u$ is an induced $P_4 + P_1$ in the graph; $u$ is non-adjacent to all of $c_2',c_1',c_1,d$ by definition of $C$ and $D$.
				\item If $(c_1', b_1)\in E$ then $\langle b_2, u, b_1, c_1' \rangle + d$ is an induced $P_4 + P_1$ in the graph; $(c_1',u)\notin E$ by definition of $C$, and $(d,b_1),~(d,b_2),~(d,u)~\notin~E$ by definition of $D$.
			\end{itemize}
		Hence, we get a contradiction, and this case cannot occur in $G$.
		
		\paragraph{(2a.2.2)} The shortest path from $d'$ to $b_2$ is of the form $\langle d', c_1', b_1', b_2 \rangle$ with $b_1' \in B_1$, $b_1'\neq b_1$.
		Note that $(c_1,u)\notin E$ by definition of $C$, $(d',b_1),(d',b_2),(d',u)~\notin~E$ by definition of $D$, and $(b_1,b_2),(c_1,b_2)\notin E$, as $(d,b_2)$ is good. But then $(d',c_1)\in E$: otherwise, $\langle c_1, b_1, u, b_2 \rangle + d'$ is an induced $P_4 + P_1$ in the graph.
	
		Say that $d$ has some other neighbour $c_1'' \in C_1$, so $(c_1'', b_1'') \in E$ for some $b_1''\in B_1$ (possibly $b_1'' = b_1$), but $(b_1'',b_2)\notin E$ because $(d,b_2)$ is good. Then, $c_1''$ can fulfil the role of $c_1$ in the above analysis, so it must be that $(d', c_1'') \in E$.
		From this analysis we can conclude that, for any $b_2 \in B_2$, if $d\in D$ is such that $(d,b_2)$ is good, and $d'\in D$ is such that $(d',b_2)$ is bad, then it must be that $N(d)\cap C_1 \subset N(d')\cap C_1$. Hence, it is only worth to consider some $d \in D$ to be good if it has minimum degree to $C$.
		
		We now have the tool to prove correctness of \textbf{(2a.algorithm)} for this case. If the BFS from the picked $d\in D$ finds a vertex at distance 4, then clearly a length-4 diametral path exists in $G$. The only risk is that we conclude there is no diametral path corresponding to this case even though it does exist. To this end, let $d\in D$ be the vertex picked by the algorithm, and assume $(d,b)$ is bad for all $b\in B_2$. Assume to the contrary that there exist $d' \in D$, $b_2\in B_2$ which are at distance 4 in a path of type \textbf{(2a)}. Then $(d', b_2)$ is good. So, by the analysis above, it must be that $N(d')\cap C_1 \subset N(d)\cap C_1$, which contradicts the assumption that $d$ was picked to have the minimal size neighbourhood with respect to $C_1$. So all pairs $(d,b_2)$ with $d\in D$, $b_2\in B_2$ must be bad, and we are correct to conclude that a diametral path of type \textbf{(2a)} does not exist.

		\paragraph{(2b)~$\langle b_1,c_1,b_2,c_2,c_3 \rangle$} As $b_1$ and $b_2$ are both adjacent to $u$, if type \textbf{(2b)} exists in $G$, $\langle b_1,u,b_2,c_2,c_3 \rangle$ is also a shortest path of distance 4 from $b_1$ to $c_3$, which is a distance-4 path of type \textbf{(2a)}. The algorithm for type \textbf{(2a)} finds a diametral path of length 4 or concludes that no diametral path of type \textbf{(2a)} exists in $G$, which also rules out that a diametral path of type \textbf{(2b)} exists in $G$.
		\paragraph{(2c)~$\langle b_1,b_2,c_1,c_2,c_3 \rangle$} We analyse the structure of the graph if type \textbf{(2c)} exists in $G$. Notice first that $c_3$ cannot be adjacent to any $b\in B$, as that would make the distance from $c_3$ to $b_1$ at most 3 via $u$. Now, either the distance from $u$ to $c_3$ is also 4, which is identified by the algorithm for type \textbf{(1a)} if it exists, or $c_2$ has a neighbour $b_3 \in B$. Note that $b_3 \neq b_1$ and $b_3 \neq b_2$, as the distance from $b_1$ to $c_3$ is 4. If $(b_1,b_3) \in E$, then the distance from $b_1$ to $c_3$ is not 4. So, $(b_1,b_3) \notin E$. But then $\langle b_1, u, b_3, c_2, c_3\rangle$ is also a shortest path from $b_1$ to $c_3$ of distance 4, which is a length-4 shortest path of type \textbf{(2a)}. The algorithm for type \textbf{(2a)} finds a diametral path of length 4 or concludes that no diametral path of type \textbf{(2a)} exists in $G$, which also rules out that a diametral path of type \textbf{(2c)} exists in $G$.
		\paragraph{(2d)~$\langle b_1,b_2,b_3,c_1,c_2 \rangle$} As all $b\in B$ are adjacent to $u$, if a shortest path of type \textbf{(2d)} exists, $\langle b_1,u,b_3,c_1,c_2 \rangle$ is also a shortest path of length 4 from $b_1$ to $c_2$, which is a length-4 shortest path of type \textbf{(2a)}. The algorithm for type \textbf{(2a)} finds a diametral path of length 4 or concludes that no diametral path of type \textbf{(2a)} exists in $G$, which also rules out that a diametral path of type \textbf{(2d)} exists in $G$.
		\paragraph{(3a)~$\langle c_1,c_2,b,c_3,c_4 \rangle$} If a diametral path of type \textbf{(3a)} exists, notice that there are no edges between $c_1,c_2$ and $c_3,c_4$, and $c_1$ is not adjacent to $b$, as otherwise the distance from $c_1$ to $c_4$ is not 4. Then, $\langle c_1, c_2, b, u\rangle$ together with $c_4$ form an induced $P_4 + P_1$ in the graph, a contradiction. So, this type cannot occur in a $(P_4+P_1)$-free graph.
		\paragraph{(3b)~$\langle c_1,b,c_2,c_3,c_4 \rangle$} We first give an algorithm and then prove its correctness.
		
		\paragraph{(3b.algorithm)} Check in $O(n + m)$ time that there exists at least one and at most two components $C_1,C_2\subseteq C$ of $G[C]$ with each at least one vertex $b \in B$ such that $b$ has an edge and a non-edge to that component, by enumerating the adjacencies of vertices $b\in B$. If this is not the case, no diametral path of type \textbf{(3b)} exists in $G$. Let $B_{C_1}, B_{C_2} \subseteq B$ be the set of vertices with an edge and a non-edge to $C_1$ and $C_2$, respectively, which we can find in $O(n+m)$ time. If only one such component exists, execute the following steps only for $C_1$ and $B_{C_1}$. For $i\in{1,2}$, mark all vertices in $C_i$ with a neighbour in $B_{C_i}$, and then mark all vertices in $C_i$ with a marked neighbour. This takes $O(n+m)$ time. Pick an arbitrary unmarked vertex in $C_1$ and one in $C_2$, and execute a BFS from each. Return a distance-4 shortest path if it is found. Otherwise, we conclude that no diametral path of type \textbf{(3b)} exists in $G$.
		
		\begin{figure}[t]
			\centering
			\includegraphics{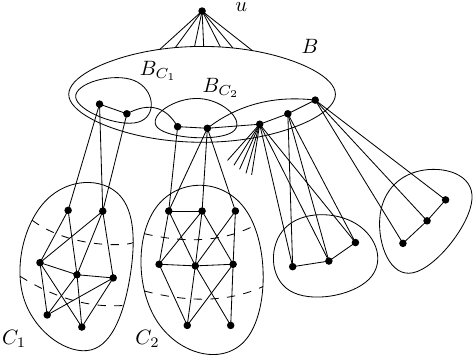}
			\caption{Structure for type \textbf{(3b)}, where for $i\in \{1,2\}$, vertices in $B_{C_i}$ are complete to all components of $G[C]$ except for $C_i$, due to Claim~\ref{clm:completetoallC}. $C_1$ and $C_2$ have layers corresponding to which vertices are marked by the algorithm, with the bottom layer consisting of unmarked vertices.
			}\label{fig:Case3bStructure}
		\end{figure}
		\medskip
		
		We next prove the correctness of this algorithm. See Figure~\ref{fig:Case3bStructure} for an illustration.
		If there is no component like $C_1$ or $C_2$, no shortest path of the form $\langle c_1,b,c_2,c_3,c_4 \rangle$ exists, as it must be that $(b,c_2), (c_2,c_3) \in E$ while $(b,c_3) \notin E$ in the diametral path of this type. If there are at least three different components of $G[C]$, $C_1,C_2,C_3\subseteq C$, which each have a vertex in $B$ with an edge and a non-edge to that component, let these be $b_{C_1},b_{C_2},b_{C_3}$ respectively. The distance between any two $c_1,c_2\in C$ is now at most 2, as by Claim~\ref{clm:completetoallC}, $b_{C_1},b_{C_2},b_{C_3}$ are complete to all components except $C_1,C_2,C_3$, respectively. So, each pair $c_1,c_2\in C$ has at least one common neighbour among $b_{C_1},b_{C_2},b_{C_3}$. So a diametral path of this type cannot occur in $G$. Hence, we are correct to look for at least one and at most two such components. We also see that one end of the diametral path must be in $C_1$ or $C_2$, if it exists. 
		
		For $i\in \{1,2\}$, all vertices in $C_i$ marked by the algorithm have distance at most 3 to all $c\in C$: vertices in the same component are at distance at most~3, as $G[C]$ is $P_4$-free. Next to this, marked vertices are at distance at most 2 from some vertex in $B_{C_i}$, and by Claim~\ref{clm:completetoallC} vertices in $B_{C_i}$ are complete to all other components of $G[C]$. So, the only candidates in $C_1$ or $C_2$ for one end of the distance-4 diametral path are the unmarked vertices. Assume that the algorithm does not find a distance-4 shortest path, but a diametral path of type \textbf{(3b)} does exist in $G$. Assume one end is in $C_1$, and let $c_1 \in C_1$ be the vertex the algorithm picked in $C_1$. The case where one end is in $C_2$ is analogous. Note that all vertices not in $B_{C_1}$ are either complete or anti-complete to $C_1$. Also, unmarked vertices in $C_1$ are at distance exactly 2 from all vertices in $C_1\cap N(B_{C_1})$, as $G[C]$ is $P_4$-free. We see that all unmarked vertices in $C_1$ have equal distance to each $b\in B$, and in particular, to the $b\in B$ the diametral path uses. Also, as the other end of the diametral path is in another component, the distance to that vertex is equal for every unmarked vertex in $C_1$. But then the BFS from $c_1$ would have found a distance-4 shortest path, a contradiction. We can conclude the algorithm is correct, that is, if the algorithm does not find a distance-4 shortest path, no diametral path of type \textbf{(3b)} exists in $G$.
		\paragraph{(3c)~$\langle c_1,b_1,u,b_2,c_2 \rangle$} First note that $G[C]$ must consist of multiple connected components for a diametral path of type \textbf{(3c)} to occur: a connected $P_4$-free graph has diameter at most 2, so the shortest path between $c_1,c_2\in C$ will never go through $u$.
		
		By Claim~\ref{clm:completetoallC} we have that any vertex in $B$ is either complete or anti-complete to every connected component of $G[C]$, or is not complete nor anti-complete to exactly one connected component of $G[C]$ and complete to all other connected components of $G[C]$. We distinguish cases \textbf{(3c.1)} and \textbf{(3c.2)}.
		
		\paragraph{(3c.1)} Every $b\in B$ is complete or anti-complete to every component of $G[C]$. It can be checked in $O(n+m)$ time whether we are in this case. Then each component of $G[C]$ is a twin class with respect to $B$. We first give the algorithm to solve this case, and then prove its correctness.
		
		\paragraph{(3c.1.algorithm)}
		\begin{figure}
			\centering
			\includegraphics[width=\textwidth]{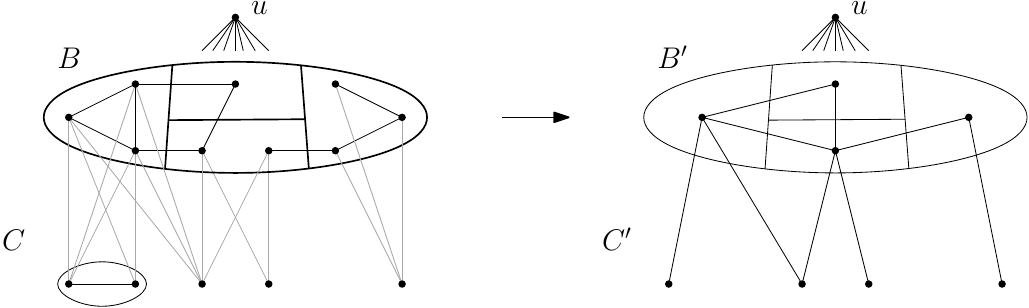}
			\caption{An illustration of the conversion of a graph $G$ (left) to a graph $G'$ (right) as in \textbf{(3c.1.algorithm)}. Note that for image simplicity, the drawn graph $G$ is not actually a $(P_4+P_1)$-free graph.}\label{fig:Case3cStructure}
		\end{figure}
		For every component of $G[C]$, delete all-but-one vertex in $O(n+m)$ time. Let $C'$ be the resulting set of vertices.
		Identify twins within $B$ with respect to their neighbourhood in $C'$, which can be done in linear time by Corollary~\ref{cor:bipartitetwins}. We identify every `class' of vertices with the same neighbourhood in $C'$ with single vertices. Let $B'$ be the set of vertices corresponding to classes. We compute the neighbourhoods of the vertices in $B'$ as the union of all neighbourhoods of vertices in the class. In particular, the edges of a $b'\in B'$ towards $C'$ are simply the neighbourhood towards $C'$ of an arbitrary $b\in B$ from that class, as every vertex in a class has an identical neighbourhood to $C$. To compute the adjacency list of vertices in $B'$, first initialize an array on the set of vertices of $B'$, in linear time. We use the array to avoid adding duplicate edges to adjacency lists, but only initialize it once. For each class $b_1' \in B'$, associate a unique label with it, and enumerate the edge lists of all $b_1\in B$ in the corresponding class. For an edge $(b_1,b_2)$, $b_2\in B$ with class $b_2'\in B'$, do the following. If the array does not contain the unique label of $b_1'$ at location $b_2'$, overwrite the location $b_2'$ with the unique label of $b_1'$, and add $b_2'$ to the adjacency list of $b_1'$. If the array does contain the label of $b_1'$ at location $b_2'$, continue to the next edge. We see that we enumerate all edges in $G[B]$ at most twice, and do constant-time operations for each one. Hence, this procedure takes $O(n+m)$ time. Call the resulting graph $G' = (V',E')$ with $B', C'\subseteq V'$. An illustration of the conversion of a graph $G$ to $G'$ is given in Figure~\ref{fig:Case3cStructure}. In $G'$, check that $G'[B']$ consists of two non-empty disjoint cliques $B_1,B_2$, where $B_1$ and $B_2$ are anti-complete, with possibly a third clique $X$ complete to $B'\setminus X$. This check takes $O(n+m)$ time by inspecting the neighbourhoods of all vertices in $B'$. If $G'[B']$ does not have this structure, return that there is no diametral path of type \textbf{(3c)}. If it does, return that a diametral path of type \textbf{(3c)} exists if there are vertices $c_1,c_2\in C'$ with $N(c_1) = B_1$ and $N(c_2) = B_2$. This can be checked in linear time by inspecting the edge lists of vertices $c\in C'$.
		\medskip

		We next prove the correctness of the algorithm. We do this by proving several properties of $G'$, to then analyse the structure of $G'$ using these properties.
		\begin{claim}\label{clm:GG'equiv2}
			A diametral path of the form $\langle c_1,b_1,u,b_2,c_2 \rangle$ exists in $G$ if and only if a diametral path of the form $\langle c_1',b_1',u,b_2',c_2' \rangle$ exists in $G'$, with $c_1',c_2'\in C'$ and $b_1',b_2'\in B'$.
		\end{claim}
		\begin{proof}
			First, note that the process to get $G'$ from $G$ can be modelled as edge contractions and vertex identifications, which means the diameter of $G'$ cannot be larger than that of $G$, and is at most 4.

			Assume a diametral path of the form $\langle c_1,b_1,u,b_2,c_2 \rangle$ exists in $G$, with $c_1,c_2\in C$ and $b_1,b_2\in B$. 
			Then the distance from $c_1$ to $c_2$ is 4 in $G$, and $(b_1,c_2)\notin E$, $(b_2,c_1)\notin E$, $(b_1,b_2)\notin E$, and $c_1,c_2$ do not have a common neighbour in $G$. 
			Let $b_1',b_2'\in B'$ be the corresponding vertices of $b_1,b_2$ respectively in $G'[B']$, and let $c_1',c_2'\in C'$ be the vertices in the components of $c_1,c_2$ in $G[C]$ that are in $C'$. 
			If $c_1$ and $c_2$ are in the same component of $G[C]$, then the distance between $c_1$ and $c_2$ is less than $4$ because $G[C]$ is $P_4$-free. So, $c_1$ and $c_2$ must not be in the same component of $G[C]$ and hence, $c_1'\neq c_2'$. Also, $b_1'\neq b_2'$, as $b_1$ and $b_2$ do not have identical neighbourhoods w.r.t.\ $C$. 
			It also holds that $(b_1',b_2')\notin E'$ as the distance between $c_1$ and $c_2$ is 4 in $G$, and so $N(c_1)$ and $N(c_2)$ are disjoint in $G$. 
			Hence, $\langle c_1',b_1',u,b_2',c_2' \rangle$ is an induced path in $G'$, but it remains to show this is the shortest path between $c_1'$ and $c_2'$ and the distance between these vertices is 4. 
			Assume to the contrary the distance from $c_1'$ to $c_2'$ is at most 3. As $c_1'\neq c_2'$, the distance is at least 2. 
			The distance between $c_1$ and $c_2$ is 4 in $G$, so $N[c_1]$ and $N[c_2]$ are disjoint, and as every vertex in $B$ is complete or anti-complete to every component of $G[C]$, we see that $N[c_1'] $ and $N[c_2'] $ are disjoint in $G'$. 
			So the distance between $c_1'$ and $c_2'$ is 3. Then, as $G[C']$ is an independent set, there exist $b_3',b_4'\in B'$ with $(b_3',b_4')\in E'$ and $(b_3',c_1)\in E'$, $(b_4',c_2)\in E'$. 
			By the class identification process, we get that $c_1'$ is adjacent to the whole class corresponding to $b_3'$, and similarly, $c_2'$ is adjacent to the whole class corresponding to $b_4'$. 
			But then there must exist $b_3,b_4\in B$ with $(b_3,b_4)\in E$ and $(b_3,c_1)\in E$, $(b_4,c_2)\in E$, a contradiction. So the distance between $c_1'$ and $c_2'$ is also 4 in $G'$, and hence, $c_1'$ and $c_2'$ are a diametral pair in $G'$ with diametral path $\langle c_1',b_1',u,b_2',c_2' \rangle$.
			
			Assume a diametral path of form $\langle c_1',b_1',u,b_2',c_2' \rangle$ exists in $G'$, with $c_1',c_2'\in C'$ and $b_1',b_2'\in B'$. Then the distance from $c_1'$ to $c_2'$ is 4. Let $c_1,c_2\in C$ be the vertices in $G$ corresponding to $c_1'$ and $c_2'$ in $G'$. Because every vertex in $B$ is complete or anti-complete to every component of $G[C]$, the vertex removal to get to $C'$ did not increase the distance between $c_1'$ and $c_2'$ as compared to $c_1$ and $c_2$. As edge/non-edge contraction can never increase the distance between two vertices, the distance between $c_1$ and $c_2$ must be at least 4 in $G$. The distance between $c_1$ and $c_2$ cannot be 5, as $G$ is $(P_4+P_1)$-free. So the distance between $c_1$ and $c_2$ is also 4 in $G$. Let $b_1,b_2\in B$ be two arbitrary vertices of the classes corresponding to $b_1',b_2'$, respectively. Then $(c_1,b_1)\in E$ and $(c_2,b_2)\in E$, as the algorithm identified classes with identical neighbourhoods towards $C$ with single vertices and gave those vertices the same neighbourhood. Also, $(b_1,b_2) \notin E$ as $(b_1',b_2')\notin E$ and $b_1',b_2'$ have as neighbourhoods the union of all neighbourhoods in their respective class. Hence, $\langle c_1,b_1,u,b_2,c_2 \rangle$ is also a diametral path in $G$.
		\end{proof}
		
		\begin{claim}\label{clm:G'noP4+P1}
			$G'$ does not contain a $P_4+P_1$ of the form $\langle c_1',b_1',u,b_2' \rangle + c_2'$, with $c_1',c_2'\in C'$ and $b_1',b_2'\in B'$.
		\end{claim}
		\begin{proof}
			Assume to the contrary that there is such a $P_4+P_1$ in $G'$. Let it be $\langle c_1',b_1',u,b_2' \rangle + c_2'$, where $b_1',b_2'\in B'$ and $c_1',c_2'\in C'$.
			Every component in $G[C]$ is a twin class with respect to $B$, and all vertices in $B$ that get identified with the same vertex in $B'$ have identical neighbourhoods with respect to $G[C]$ and $G[C']$. Hence, the identifications from $G[B]$ to $G[B']$ essentially only adjusted adjacencies within $G[B]$ and removed vertices from $B$. Denote $b_1,b_2$ two vertices in $G$ that end up identified with $b_1',b_2'$ respectively. We immediately have $(c_1,b_1), (b_1,u), (u,b_2)\in E$ and $(c_1,b_2)\notin E$. As $(b_1',b_2')\notin E$, we get that $(b_1,b_2)\notin E$, as the identifications take the union of neighbourhoods. So $\langle c_1,b_1,u,b_2 \rangle$ is an induced $P_4$ in $G$. We also see that $c_2$ is disjoint from this $P_4$, as $c_1'\neq c_2'$ and so $c_1$ and $c_2$ are in different components, and by the previous arguments, $(c_2,b_1), (c_2,b_2)\notin E$. We get an induced $P_4+P_1$ in $G$, a contradiction.
		\end{proof}
		
		\begin{claim}\label{clm:G'nonedgeinB'}
			For every non-edge $(b_1',b_2')\notin E'$ in $G'$, with $b_1',b_2'\in B'$, it holds that $b_1',b_2'$ together dominate $C'$, that is, every $c'\in C'$ is adjacent to $b_1'$ or $b_2'$.
		\end{claim}
		\begin{proof}
			Let $(b_1',b_2')\notin E'$ be an arbitrary non-edge in $G'$. As $b_1'\neq b_2'$, it must be that $N(b_1')\cap C' \neq N(b_2') \cap C'$; otherwise, these two vertices would be the same class. $C'$ is an independent set by construction.
			
			Because $b_1'$ and $b_2'$ have unequal neighbourhoods towards $C'$, at least one of the two has a neighbour in $C'$ which is not a neighbour of the other. Without loss of generality, let this be $b_1'$. So, $b_1'$ has a neighbour $c_1'\in C'$ that is not a neighbour of $b_2'$. Because $(b_1',b_2')\notin E$ and $(b_2',c_1')\notin E$, we have that $\langle c_1', b_1', u, b_2' \rangle$ is an induced $P_4$ in $G'$. Hence, every other $c' \in C'$ must be adjacent to some vertex of the $P_4$; otherwise, this $c'$ forms an induced $P_4 + P_1$ which cannot be present in $G'$ by Claim~\ref{clm:G'noP4+P1}. The neighbour of any $c'$ cannot be $c_1'$, as $C'$ is an independent set. The neighbour of any $c'$ cannot be $u$, because of the definition of $C'$. So every other $c'\in C'$ is adjacent to $b_1'$ or $b_2'$.
		\end{proof}
		
		By Claim~\ref{clm:GG'equiv2} we can look for the diametral path of type \textbf{(3c)} in $G'$ to determine its presence in $G$. To this end, call a non-edge $(b_1,b_2)\notin E'$ in $G'[B']$ `good' if a diametral path of type \textbf{(3c)} exists with $b_1,b_2$ as its vertices in $B'$. We analyse the structure of $G'[B']$ assuming a good non-edge exists.
		
		\begin{claim}\label{clm:G'B'structure}
			In $G'$, if a good non-edge $(b_1,b_2)\notin E'$ exists, then (a) no non-edge triangle is contained in $G'[B']$ as an induced subgraph, (b) every vertex $b_3\in B'\setminus \{b_1,b_2\}$ is adjacent to at least one of $b_1,b_2$, and (c) any common neighbour $x\in B'$ of $b_1$ and $b_2$ is complete to $B'\setminus \{x\}$.
		\end{claim}
		\begin{proof}
			Let $(b_1,b_2)\notin E'$ be a good non-edge, for $b_1,b_2\in B'$.
			
			To prove (a), assume there is a non-edge triangle formed by the vertices $a_1,a_2,a_3\in B'$. Apply Claim~\ref{clm:G'nonedgeinB'} to all three of the non-edges in this non-edge triangle to see that every vertex in $C$ is adjacent to at least two of $a_1,a_2,a_3$. But then $(b_1,b_2)$ is not good, as any two $c_1,c_2\in C'$ have a common neighbour among $a_1,a_2,a_3$. This is a contradiction with the assumption that $(b_1,b_2)$ is good.
			
			To prove (b), assume there is another vertex $b_3\in B'$. If $b_3$ is not a neighbour of both $b_1$ and $b_2$, we immediately have a non-edge triangle formed by the non-edges $(b_1,b_2),(b_1,b_3),(b_2,b_3)\notin E'$, a contradiction.

			To prove (c), assume $b_1,b_2$ have a common neighbour $x \in B'$. Assume for sake of contradiction that there is another vertex $b_3 \in B'$ not adjacent to $x$. By (b), $b_3$ is adjacent to at least one of $b_1,b_2$.\\
			If $b_3$ is adjacent to only one of $b_1,b_2$, say $(b_1,b_3)\in E'$, then
			by applying Claim~\ref{clm:G'nonedgeinB'} to $(x,b_3)\notin E'$ we see that every vertex of $C' \setminus N(b_1)$ is also adjacent to either $x$ or $b_3$. But then $b_1$ has distance at most 2 to every $c\in C'$, so $(b_1,b_2)$ is not good, a contradiction.\\
			If instead $b_3$ is adjacent to $b_1,b_2$ but not $x$, then, similar to before, applying Claim~\ref{clm:G'nonedgeinB'} to the non-edge $(x,b_3)$ gives us that $b_1$ has distance at most 2 to every $c\in C'$. This is a contradiction with the assumption that $(b_1,b_2)$ is good.\\
			We get that $b_3$ must be adjacent to $x$. This holds for every $b_3\in B'\setminus \{x,b_1,b_2\}$, so $x$ is complete to $B'\setminus \{x\}$, and this holds for every common neighbour $x\in B'$ of $b_1$ and $b_2$.
		\end{proof}
		
		By Claim~\ref{clm:G'B'structure}, $G'[B']$ admits very strong structural properties if a good non-edge exists. Assume a good non-edge $(b_1,b_2)\notin E'$ exists in $G'[B']$, for $b_1,b_2\in B'$. Denote by $X\subseteq B'$ the set of vertices of $B'$ complete to $B'$, and note that this is the same $X$ as the $X$ that appears in the algorithm. We further analyse the structure of $G'$.
		
		Denote by $B_1 \subseteq B'$ ($B_2\subseteq B'$) the subset of $B'$ adjacent to $b_1,X$ ($b_2,X$) but not $b_2$ ($b_1$), including $b_1$ ($b_2$). Note that, by Claim~\ref{clm:G'B'structure}~(b) and (c), $X,B_1,B_2$ partition $B'$.
		Say that for $b_3\in B_1$ and $b_4\in B_2$ we have $(b_3,b_4)\in E'$. Because we assumed $(b_1,b_2)\notin E'$, we know that $b_3,b_4$ are not both equal to $b_1,b_2$. If exactly one of $b_3 = b_1$ and $b_2=b_4$ holds, then the other of $b_3,b_4$ is adjacent to both $b_1$ and $b_2$, which is a contradiction with $b_3,b_4\notin X$. We may conclude that $b_3\neq b_1$ and $b_4\neq b_2$.
		By construction of $B_1$ and $B_2$, $(b_2, b_3)\notin E'$ and $(b_1,b_4)\notin E'$.
		By applying Claim~\ref{clm:G'nonedgeinB'} to $(b_2, b_3)\notin E'$, we get that any non-neighbour of $b_2$ in $C'$ is adjacent to $b_3$. By applying Claim~\ref{clm:G'nonedgeinB'} to $(b_1,b_4)\notin E'$, get that any non-neighbour of $b_1$ in $C'$ is adjacent to $b_4$. We see that the neighbourhood of any $c\in C'$ contains at least one of the sets $\{b_1,b_3\}$, $\{b_2,b_4\}$, and $\{b_1,b_2\}$. But because $(b_1,b_3), (b_2,b_4), (b_3,b_4)\in E'$, we get that the distance between any two vertices in $C'$ is at most 3. This is a contradiction with the assumption that $(b_1,b_2)$ is good. We see that there are no edges between $B_1$ and $B_2$.
		Any non-edge within $G'[B_1]$ or $G'[B_2]$ creates a non-edge triangle with $b_2$ or $b_1$, respectively, so then $(b_1,b_2)$ is not good by Claim~\ref{clm:G'B'structure}~(a).
		So, if we have a good non-edge $(b_1,b_2)$, then $B'$ must be a partition of three cliques $B_1$, $B_2$, $X$, where $X$ is complete to $B'\setminus X$ and $B_1,B_2$ have no edges between them.
	
		We can apply Claim~\ref{clm:G'nonedgeinB'} to every pair in $B_1\times B_2$ to see that for every $c\in C'$ it holds that $N(c) \supseteq B_1$ or $N(c) \supseteq B_2$. Then any vertex $c\in C'$ with $N(c) \supsetneq B_1$ or $N(c) \supsetneq B_2$ has distance at most 3 to all vertices in $C'$, as vertices in $X$ are complete to $B'$ and $G'[B_1]$ and $G'[B_2]$ are cliques.

		Hence, if this partition exists, then indeed a diametral path of type \textbf{(3c)} exists if and only if there exist $c_1,c_2\in C$ with $N(c_1) = B_1$, $N(c_2) = B_2$, as the distance between any $b_1\in B_1$ and $b_2 \in B_2$ is $2$, and so $c_1$ to $c_2$ must have distance 4. If the partition does not exist as described, then there is no good non-edge, and no length-4 diametral path of type \textbf{(3c)} exists. The algorithm looks for exactly the partition $B_1,B_2,X$ as described and checks the existence of $c_1,c_2\in C$ with $N(c_1) = B_1$, $N(c_2) = B_2$, which is now proven to be correct.
		
		We conclude that \textbf{(3c.1.algorithm)} only returns that there is no diametral path of length 4 when there does not exist a diametral path of type \textbf{(3c)} in $G$.
		\paragraph{(3c.2)} In this case, we may assume that there exists a $b\in B$ that has an edge and a non-edge to one component of $G[C]$.
		Assume that a diametral path of type \textbf{(3c)} exists, let it be $\langle c_1, b_1, u, b_2, c_2 \rangle$ for $c_1,c_2\in C$, $b_1,b_2\in B$. Let $b\in B$ be a vertex with an edge and a non-edge to one component of $G[C]$, denote $C_1$ the vertices of this component. By Claim~\ref{clm:completetoallC}, $b$ is complete to all components of $G[C \setminus C_1]$. Now at least one of $c_1,c_2$ must be in $C_1$, as any pair of vertices from $C\setminus C_1$ are at distance at most 2 due to $b$. It cannot be that both $c_1\in C_1$ and $c_2 \in C_1$, because $c_1$ and $c_2$ are at distance 4, and the component $G[C_1]$ must be $P_4$-free. Say that $c_1 \in C_1$ and $c_2 \notin C_1$. We get that $(b,c_2)\in E$. As the shortest path between $c_1$ and $c_2$ has length 4, it must be that $c_1$ is at distance at least 3 from $b$. $G[C_1]$ is $P_4$-free and $b$ has an edge to some vertex of $C_1$, so $c_1$ has to be at distance exactly 3 from $b$. But then there are vertices $c_3,c_4\in C_1$ such that $\langle c_1, c_3, c_4, b, c_2 \rangle$ is an induced shortest path from $c_1$ to $c_2$. This corresponds to a diametral path of type \textbf{(3b)}. The algorithm for type \textbf{(3b)} finds a diametral path of length 4 or concludes that no diametral path of type \textbf{(3b)} exists in $G$, which also rules out that a diametral path of type \textbf{(3c)} exists in $G$.
		\paragraph{(3d)~$\langle c_1,b_1,c_2,b_2,c_3 \rangle$} As $b_1$ and $b_2$ are both adjacent to $u$, if a diametral path of type \textbf{(3d)} exists, $\langle c_1,b_1,u,b_2,c_2 \rangle$ is also a shortest path of distance 4 from $c_1$ to $c_2$ which is of type \textbf{(3c)}. The algorithm for type \textbf{(3c)} finds a diametral path of length 4 or concludes that no diametral path of type \textbf{(3c)} exists in $G$, which also rules out that a diametral path of type \textbf{(3d)} exists in $G$.
		\paragraph{(3e)~$\langle c_1,c_2,b_1,b_2,c_3 \rangle$} Notice that in a diametral path of type \textbf{(3e)}, $c_3$ has no edges to either $c_1,c_2,b_1$, and $c_1$ is not adjacent to $b_1$; otherwise, the distance from $c_1$ to $c_3$ is not 4. Now $\langle c_1,c_2,b_1,u\rangle$ and $c_3$ form a $P_4+P_1$. So, this case cannot occur in a $(P_4+P_1)$-free graph.
		\paragraph{(3f)~$\langle c_1,b_1,b_2,b_3,c_2 \rangle$} As $b_1$ and $b_3$ are both adjacent to $u$, if a diametral path of type \textbf{(3f)} exists in $G$, $\langle c_1,b_1,u,b_3,c_2 \rangle$ is also a shortest path of distance 4 from $c_1$ to $c_2$ which is of type \textbf{(3c)}. The algorithm for type \textbf{(3c)} finds a diametral path of length 4 or concludes that no diametral path of type \textbf{(3c)} exists in $G$, which also rules out that a diametral path of type \textbf{(3f)} exists in $G$.
		\paragraph{} We have now shown that, for each possible structure of the diametral path with respect to $u,B,C$, we find no diametral path of length 4 only when that structure does not exist in $G$. So, we have shown that in time $O(n+m)$ we find a distance-4 shortest path if and only if it exists in $G$.
	\end{proof}

	\subsection{\texorpdfstring{$(P_3+2P_1)$}{(P3+2P1)}-free graphs}
	
	\begin{figure}[t]
		\centering
		\includegraphics[width=.6\textwidth]{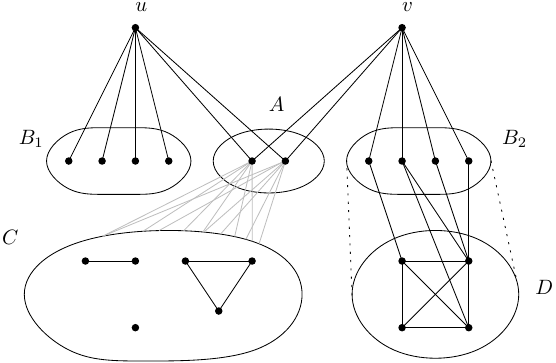}
		\caption{An illustration of the structure of a $(P_3 + 2P_1)$-free graph as in the proof of Theorem~\ref{thm:P3+2P1free}.}\label{fig:P3+2P1structure}
	\end{figure}
	
	\begin{theorem}\label{thm:P3+2P1free}
		Given a $(P_3 + 2P_1)$-free graph $G$, we can decide whether the diameter of $G$ is equal to $d_{\max} = 5$ in $O(n+m)$ time.
	\end{theorem}
	\begin{proof}
		Let $G = (V,E)$ be a (connected) $(P_3 + 2P_1)$-free graph. Indeed, $d_{\max} = 5$ by Theorem~\ref{thm:HLongestPath}. We start with a BFS from an arbitrary vertex $u$, and let $v$ be a vertex that is in the second neighbourhood of $u$. If such a vertex $v$ does not exist, the diameter of $G$ is not 5. We now distinguish the structure of the graph with respect to $u$ and $v$. This structure can be identified by BFS from both $u$ and $v$, see also Figure~\ref{fig:P3+2P1structure} for an illustration. Let $A \subset V$ be the set of all common neighbours of $u$ and $v$. Note that $A \neq \emptyset$ by the choice of $v$. Set $B_1 = N(u)\setminus A$ and $B_2 = N(v)\setminus A$. For now, let $C' \subset V$ be the set of all other vertices in the graph. The distances between vertices in $V\setminus C'$ are at most 4 due to $u$ and $v$. So, if $C'$ is empty, we are done. Otherwise, any diametral path of length 5 has an endpoint in $C'$, so we focus on this. Note that $G[C']$ is $P_3$-free, a disjoint union of cliques, as it is non-adjacent to both $u$ and $v$.
		For any $a\in A$, the non-neighbours of $a$ in $C'$ must form a complete subgraph, as $\langle u, a, v \rangle$ is a $P_3$ and the graph is $(P_3 + 2P_1)$-free. We see that any $a\in A$ is non-adjacent to at most one clique or a part of one clique in $C'$.
		Say that every clique in $C'$ has a vertex with a neighbour in $A$. Consider two vertices $c_1',c_2'\in C'$. If $c_1'$ and $c_2'$ are in the same clique, then the distance from $c_1'$ to $c_2'$ is $1$. Otherwise, let $C_1\subset C'$ ($C_2\subset C'$) be the vertices of the clique containing $c_1$ ($c_2$). If there is an $a\in A$ that is connected to both cliques $C_1$ and $C_2$, then the distance from $c_1'$ to $c_2'$ is at most $4$. Otherwise, $C_1$ is connected to a vertex $a_1\in A$ and $C_2$ is connected to a vertex $a_2\in A$, where $a_1\neq a_2$. But because $a_1$ is not connected to $C_2$ by assumption, it holds that $(c_1,a_1)\in E$ because any vertex of $A$ is non-adjacent to at most one part of one clique. Similarly, $(c_2,a_2)\in E$. But then the distance from $c_1'$ to $c_2'$ is at most $4$. We see that the distances in $G$ are at most 4 if every clique in $C'$ has a vertex with a neighbour in $A$.
		So, for the diameter to be 5, there is a clique of vertices $D \subseteq C'$ that is non-adjacent to all of $A$. Define $C$ as $C = C' \setminus D$. 
	
		Every vertex $a \in A$ is complete to $C$ by construction. We can identify whether such a clique $D$ exists, and find all vertices in $D$ if it does, in linear time, by inspection of the adjacency lists of all vertices in $A$. If the distance from $u$ or $v$ to some vertex is 5, the initial BFS from $u$ or $v$ would identify such a diametral path. Vertices in $C$ have distance at most 4 to all vertices in $V \setminus D$, because they are complete to $A$.

		As an algorithm, we claim that it now suffices to find a minimum-degree vertex in $D$ and execute a BFS from it. Return a length-5 diametral path if it finds a vertex at distance 5 and otherwise, return that no length-5 diametral path exists in $G$.
	
		We can see that one end of any diametral path of length 5 has to be in $D$, if such a path exists. The other end is either a vertex in $C$ or a vertex in $B_1\cup B_2\cup A$. Note that some vertex of $D$ has an edge to either $B_1$ or $B_2$, because $G$ is connected. Hence, the distance from any $d\in D$ to any $a\in A$ is at most 4.
		
		For a $d\in D$ to have distance 5 to some $c\in C$, it must be that $N(d) \subseteq D$, as $u$ and $v$ have distance 2 to all $c\in C$. Note that all $d\in D$ with $N(d) \subseteq D$ are true twins, as $D$ is a clique. So a single BFS from the lowest degree vertex in $D$ will identify a diametral path of length 5 from $D$ to $C$, if it exists.
		
		In the following, assume the distance from any $d\in D$ to all $c\in C$ is at most 4.
		If there are $d_1,d_2 \in D$ such that $N(d_1) \cap B_1 \neq \emptyset$ and $N(d_2) \cap B_2 \neq \emptyset$, then the distance from any $d\in D$ to any other vertex in the graph is at most 4, as the distance from any $d\in D$ is at most 3 to both $u$ and $v$. We can check in linear time whether $D$ has some edge to $B_1$ and/or some edge to $B_2$, and terminate when there is an edge for both. Without loss of generality, assume that vertices in $D$ are not adjacent to $B_1$. Now a diametral path of length 5 can only go from $D$ to $B_1$, if the diameter is 5.
		
		Then, we must have $B_1, B_2 \neq \emptyset$. If some $c\in C$ has a neighbour $b_1\in B_1$, then $\langle u,b_1,c \rangle + v + d$ forms a $P_3 + 2P_1$, where $d$ is any vertex in $D$. So, either $C$ is empty, or no vertex in $C$ has a neighbour in $B_1$. Regardless, we do not have to consider finding any shortest path from a vertex in $D$ to any vertex in $B_1$ that goes through $C$, as using $v$ instead is equivalent. Let $B_2'\subseteq B_2$ be the set of vertices in $B_2$ that each have at least one neighbour in $A$ or $B_1$. Let $D'\subseteq D$ be the set of vertices in $D$ with at least one neighbour in $B_2'$.
		Vertices in $D'$ have distance 3 to $u$, and so distance at most 4 to all of $B_1$. All vertices in $D\setminus D'$ have distance at least 4 to $u$. If every $b_1 \in B_1$ is complete to $A$, then any shortest path of length 5 from a vertex in $D$ to a vertex in $B_1$ would need its endpoint $d\in D$ to have $N(d)\subseteq D$. All such vertices are true twins, and a BFS from a single vertex in $D$ with minimum degree identifies such a diametral path, if it exists. Otherwise, there exist $b_1\in B_1$, $a\in A$ with $(b_1,a) \notin E$. But then $G[(B_2\setminus B_2')\cup D]$ must be both $P_3$-free and $2P_1$-free, as $b_1$ and $a$ form a $2P_1$ but also a $P_3$ with $u$. So, $G[(B_2\setminus B_2')\cup D]$ is a clique. But then all $d\in D\setminus D'$ are true twins, and have lower degree than the vertices in $D'$. A BFS from any of these minimum-degree vertices in $D$ suffices to find a length-5 diametral path, if it exists.
		
		The above analysis shows that we can find a certificate for a shortest path of length 5 if and only if it exists in $G$.
	\end{proof}
	
	As an immediate corollary, we get that we can decide on $d_{\max}$ for the class of $4P_1$-free graphs, as any $4P_1$-free graph is $(P_3 + 2P_1)$-free, and $d_{\max}$ is equal for both classes.

	\begin{corollary}\label{cor:4P1free}
		Given a $4P_1$-free graph $G$, we can decide whether the diameter of $G$ is equal to $d_{\max} = 5$ in $O(n+m)$ time.
	\end{corollary}
	
	\subsection{\texorpdfstring{$(P_2+3P_1)$}{(P2+3P1)}-free graphs}
	
	\begin{figure}[t]
		\centering
		\includegraphics[width=.6\textwidth]{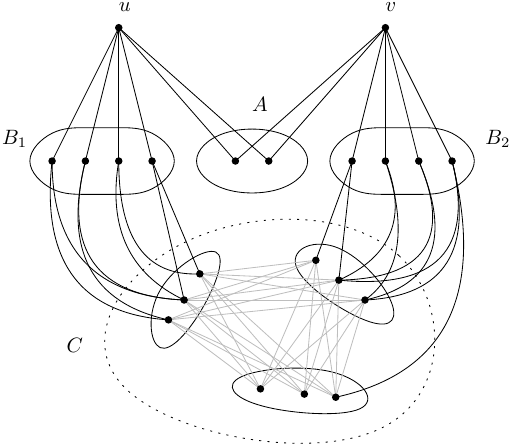}
		\caption{An illustration of the structure of a $(P_2 + 3P_1)$-free graph as in the proof of Theorem~\ref{thm:P2+3P1free}.}\label{fig:P2+3P1structure}
	\end{figure}
	
	\begin{theorem}\label{thm:P2+3P1free}
		Given a $(P_2 + 3P_1)$-free graph $G$, we can decide whether the diameter of $G$ is equal to $d_{\max} = 6$ in $O(n+m)$ time.
	\end{theorem}
	\begin{proof}
		Let $G = (V,E)$ be a (connected) $(P_2 + 3P_1)$-free graph. Indeed, $d_{\max} = 6$ by Theorem~\ref{thm:HLongestPath}. We start with a BFS from an arbitrary vertex $u$, and let $v$ be a vertex that is in the second neighbourhood of $u$. If such a vertex $v$ does not exist, the diameter of $G$ is not 6. We now distinguish the structure of the graph with respect to $u$ and $v$. This structure can be identified by BFS from both $u$ and $v$, see also Figure~\ref{fig:P2+3P1structure} for an illustration. Let $A \subset V$ be the set of common neighbours of $u$ and $v$. Note that $A \neq \emptyset$ by the choice of $v$. Set $B_1 = N(u)\setminus A$ and $B_2 = N(v)\setminus A$. Let $C \subset V$ be the set of all other vertices in the graph. If $C$ is empty, the diameter of $G$ cannot be 6. As $u + v$ is a $2P_1$, $G[C]$ must be $(P_2 + P_1)$-free, a complete $r$-partite graph for some integer $r\geq 1$.
		
		\begin{claim}\label{clm:P2+3P1NNC}
			The non-neighbours in $C$ of a vertex $b\in B_1\cup B_2$ form a complete subgraph.
		\end{claim}
		\begin{proof}
			Say $b\in B_1$. Then $\langle b,u \rangle + v$ is a $P_2 + P_1$, so the induced subgraph of vertices non-adjacent to all of $\{b,u,v\}$ is $2P_1$-free. This is the set of vertices in $C$ non-adjacent to $b$. The case for some $b \in B_2$ is analogous.
		\end{proof}
		
		If the distance from $u$ or $v$ to some vertex is 6, then the initial BFS would identify such a diametral path. Otherwise, at least one endpoint of a length-6 diametral path must be in $C$. We distinguish algorithmically if $r=1$ or $r>1$ by checking whether $G[C]$ contains an edge in $O(n+m)$ time.
		
		If $G[C]$ is $1$-partite, an independent set, then every $c\in C$ has a neighbour in $B_1\cup A \cup B_2$ due to $G$ being a connected graph. Because of this, the distance from any $c\in C$ to any vertex in $B_1\cup A\cup B_2 \cup \{u,v\}$ is at most 5. So, a diametral path of length 6 then must have both endpoints in $C$, if it exists.
		However, by Claim~\ref{clm:P2+3P1NNC}, the non-neighbours in $G[C]$ of any $b_1\in B_1$ or $b_2\in B_2$ form a clique, they must be a single vertex. So, after inspecting the neighbourhood of a single $b_1\in B_1$ there is at most single vertex $c_1\in C$ which is a non-neighbour of $b_1$. But then the distances between all vertices in $C\setminus \{c_1\}$ are at most 2, and $c_1$ must be an endpoint of a length-6 diametral path, if it exists. We can find $c_1$ in time $O(n+m)$ and execute a BFS from it and return a distance-6 shortest path if it is found.
		
		If $G[C]$ is not an independent set, $G[C]$ is a complete $r$-partite graph for $r > 1$ and the distance between vertices in $C$ is at most 2. The distance from some $a\in A$ to all vertices of $G$ is at most 5, as it has distance 2 to all vertices in $B_1 \cup B_2$. Then, a diametral path of length 6 must have one endpoint in $C$ and one endpoint in either $B_1$ or $B_2$. Without loss of generality, assume that one end of the diametral pair is in $B_1$, the other case is symmetric.
		
		Suppose that $G[C]$ is not a clique, and recall that $G[C]$ is a complete $r$-partite graph. Because the non-neighbours of any $b_1\in B_1$ in $C$ must form a clique by Claim~\ref{clm:P2+3P1NNC}, every $b_1\in B_1$ has a neighbour in $C$. Hence, the distance from any $b_1$ to any $c\in C$ is at most 3. Hence, the diameter of $G$ is at most 5. So, if there is a length-6 diametral path from some $c\in C$ to some $b_1\in B_1$, $G[C]$ is a clique, and there are no edges between $C$ and $B_1\cup A$. This can be checked in $O(n+m)$ time. Also, the endpoint of the diametral path in $C$, say $c\in C$, must have $N(c) \subseteq C$, as otherwise it has distance at most 5 to all other vertices. All such vertices are true twins, as $G[C]$ is a clique. So we can execute a BFS from any one of these vertices to find a length-6 shortest path with an endpoint in $C$ and one in $B_1$, if it exists. This is simply the lowest degree vertex in $C$, so finding it and executing a BFS from it takes $O(n+m)$ time.
		
		The above analysis shows that we can find a certificate for a shortest path of length 6 in $O(n+m)$ time if and only if it exists in $G$.
	\end{proof}

	\subsection{\texorpdfstring{$(2P_2+P_1)$}{(2P2+P1)}-free graphs}
	
	\begin{figure}[t]
		\centering
		\includegraphics[width=.6\textwidth]{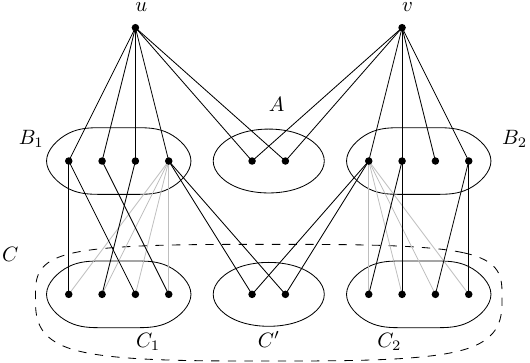}
		\caption{An illustration of the structure of a $(2P_2 + P_1)$-free graph as in the proof of Theorem~\ref{thm:2P2+P1free}.}\label{fig:2P2+P1structure}
	\end{figure}
	
	\begin{theorem}\label{thm:2P2+P1free}
		Given a $(2P_2 + P_1)$-free graph $G$, we can decide whether the diameter of $G$ is equal to $d_{\max} = 5$ in $O(n+m)$ time.
	\end{theorem}
	\begin{proof}
		Let $G = (V,E)$ be a (connected) $(2P_2 + P_1)$-free graph. Indeed, $d_{\max} = 5$ by Theorem~\ref{thm:HLongestPath}. We start with a BFS from an arbitrary vertex $u$, and let $v$ be a vertex that is in the second neighbourhood of $u$, with $N(v) \neq N(u)$. If such a vertex $v$ does not exist, the diameter of $G$ is not 5, but at most 2. Such a vertex $v$ can be identified in time $O(n+m)$. We now distinguish the structure of the graph with respect to $u$ and $v$. This structure can be identified by BFS from both $u$ and $v$; see Figure~\ref{fig:2P2+P1structure} for an illustration. Let $A \subset V$ be the set of common neighbours of $u$ and $v$. Note that $A \neq \emptyset$ by the choice of $v$. Set $B_1 = N(u)\setminus A$ and $B_2 = N(v)\setminus A$. Let $C \subset V$ be the set of all other vertices in the graph.
		
		The algorithm for this problem is given by the following steps. Search for a vertex $x_1\in C$ with $N(x_1) = B_1$, a vertex $x_2\in C$ with $N(x_2) = B_2$, and a vertex $x_3\in C$ of minimum degree among those vertices $c\in C$ with $N(c)\subseteq A$. For each $i\in \{1,2,3\}$, if it exists, execute a BFS from $x_i$, returning a distance-5 shortest path when found. If none of the options exist or none lead to a distance-5 shortest path, return that the diameter of $G$ is not 5. Notice that this can be executed in $O(n+m)$ time, as we can identify the vertices in $O(n+m)$ time, and we do a constant number of BFS calls.
		
		Next, we prove correctness of this algorithm.
		We first prove two claims that will be useful in analysing the structure of $G$.
		
		\begin{claim}\label{clm:2P2+P1_b1b2dominate}
			For all $b_1 \in B_1$, $b_2\in B_2$, if $(b_1,b_2)\notin E$ then $b_1,b_2$ together dominate all vertices in $C$.
		\end{claim}
		\begin{proof}
			Let $b_1 \in B_1$, $b_2\in B_2$ such that $(b_1,b_2)\notin E$. Assume to the contrary there is a vertex $c\in C$ that is not adjacent to either $b_1$ or $b_2$. By construction, $c$ is not adjacent to $u$ or $v$, and $(u,v)\notin E$. But then we have that $\langle u,b_1 \rangle + \langle v,b_2 \rangle + c$ is an induced $2P_2+P_1$ in the graph, a contradiction.
		\end{proof}
		
		\begin{claim}\label{clm:2P2+P1_edgeCattachedB}
			For any edge $(c_1,c_2)$ in $G[C]$, $c_1,c_2\in C$, for every $b\in B_1\cup B_2$, it holds that $(b,c_1)\in E$ or $(b,c_2)\in E$.
		\end{claim}
		\begin{proof}
			Let $(c_1,c_2)$ be an edge in $G[C]$, $c_1,c_2\in C$. Assume to the contrary that there exist a $b\in B_1\cup B_2$ with $(b,c_1)\notin E$ and $(b,c_2)\notin E$. Let us assume $b \in B_1$, the other case is symmetric. By construction, $c_1,c_2$ are not adjacent to $u$ or $v$, and $(u,v)\notin E$. But then we have that $\langle u,b \rangle + \langle c_1,c_2 \rangle + v$ is an induced $2P_2+P_1$ in the graph, a contradiction.
		\end{proof}
	
		Since $u,v$ dominate $G[V\setminus C]$ and have a common neighbour, any length-5 shortest path must have at least one endpoint in $C$, as the distances in $G[V\setminus C]$ are at most 4. If the distance from $u$ or $v$ to some vertex is 5, then we would know by the initial BFS\@.
		
		\medskip
		\noindent
		\textbf{Case 1:} There exist $b_1\in B_1, b_2\in B_2$ with $(b_1,b_2)\notin E$.\\
		Let $C' = N(b_1)\cap N(b_2)$ and let $C_1 = (N(b_1) \cap C)\setminus C'$ and $C_2 = (N(b_2) \cap C)\setminus C'$. Note that $C_1, C', C_2$ together partition $C$ by Claim~\ref{clm:2P2+P1_b1b2dominate}. We immediately see that the distance from any vertex in $C'$ to any other vertex of $G$ is at most $3$. The same holds for vertices in $A$. Whenever $C' \neq \emptyset$, the distances between vertices of $C$ are at most $4$.

		We first show that we find a diametral path of length 5 with one endpoint in $B_1$ or $B_2$ and the other in $C_2$ or $C_1$, respectively, if it exists. These two scenarios are symmetric, so assume without loss of generality that a diametral path exists with endpoints $c\in C_1$ and $b_3\in B_2$. Because the distance from $c$ to $b_3$ is 5, $b_3$ is not adjacent to any vertex of $B_1$, or any vertex of $C_1\cup C'$. So, by applying Claim~\ref{clm:2P2+P1_b1b2dominate} with $b_1$ and $b_3$ we see that $N(b_3) \cap C = C_2$ as $C_2 \cap N(b_1) = \emptyset$. Furthermore, for any $b_1'\in B_1$ we get $N(b_1')\cap C \supseteq C_1 \cup C'$ by applying Claim~\ref{clm:2P2+P1_b1b2dominate} with respect to $b_1'$ and $b_3$. Also, $c$ has no neighbours in $A, C', C_2$, or $B_2$, as the distance to $b_3$ is 5. Consider some edge $e$ in $G[C_1]$. By applying Claim~\ref{clm:2P2+P1_edgeCattachedB} to $e$, we get that one of the vertices of $e$ is adjacent to $b_3$. We see that $c$ has degree 0 in $G[C_1]$. To sum up, $N(c) = B_1$, and this holds for any $c\in C$ with a vertex in $B_2$ at distance 5. Hence, all these candidates are (false) twins. We find that detection of a single $c\in C$ with $N(c) = B_1$ and executing a BFS from it indeed suffices to find a length-5 diametral path from a vertex in $C_1$ to a vertex in $B_2$, if it exists. For the symmetric case, this argument shows that it suffices to find a vertex $c\in C$ with $N(c) = B_2$ and execute a BFS from it to find a length-5 diametral path from a vertex in $C_2$ to a vertex in $B_1$, if it exists. The algorithm searches for $x_1$ and $x_2$ meeting these conditions, so it will find a length-5 diametral path from $B_1$ to $C_2$ or $B_2$ to $C_1$ if it exists.

		For this case, it remains to show that we also find a diametral path with both endpoints in $C_1 \cup C_2$, if it exists. However, we show instead that such a diametral path cannot exist. Assume for sake of contradiction that the distance from $c_1\in C_1$ to $c_2\in C_2$ is 5. It must be that $C' = \emptyset$ for the distance between $c_1$ and $c_2$ to be 5. Let it be clear that $c_1$ has no neighbour in $\{u,v\} \cup A \cup B_2 \cup C_2$ and similarly $c_2$ has no neighbour in $\{u,v\} \cup A \cup B_1 \cup C_1$, as the distance between $c_1$ and $c_2$ would not be 5. But then for any $a\in A$ it must be that $(a,b_1),(a,b_2)\in E$, as otherwise $c_1 + \langle u,a \rangle + \langle b_2,c_2 \rangle$ would be an induced $2P_2+P_1$ and $c_2 + \langle v,a \rangle + \langle b_1,c_1 \rangle$ would be an induced $2P_2+P_1$. The existence of these edges is a contradiction with the assumption that the distance from $c_1$ to $c_2$ is 5. We conclude that the distance from any $c_1\in C_1$ to any $c_2\in C_2$ cannot be 5, and that the only length-5 diametral path that can exist has one endpoint in $B_1$ or $B_2$ and the other in $C_2$ or $C_1$, respectively.

		\medskip\noindent
		\textbf{Case 2:} There do not exist $b_1\in B_1, b_2\in B_2$ with $(b_1,b_2)\notin E$.\\
		We get that $B_1$ and $B_2$ are complete to each other, or one of $B_1$ and $B_2$ is empty. By choice of $v$, it cannot be that both $B_1$ and $B_2$ are empty sets. We get that vertices in $B_1\cup B_2$ have distance at most 3 to any other vertex in the graph, as by Claim~\ref{clm:2P2+P1_edgeCattachedB} vertices incident to an edge in $G[C]$ are at distance at most 2 from any $b \in B_1\cup B_2$, and vertices not incident to an edge in $G[C]$ are adjacent to at least one vertex in $B_1\cup A \cup B_2$, and so have distance at most 3 to any vertex in $B_1\cup B_2$. Vertices in $A$ have distance at most 4 to any vertex in the graph, because any $c\in C$ is at distance at most 2 from some vertex in $B_1\cup A \cup B_2$ by Claim~\ref{clm:2P2+P1_edgeCattachedB}. If $u$ or $v$ are at distance 5 from some vertex in the graph, we would know by the original BFS\@. Hence, the only option for a distance-5 pair is two vertices in $C$, say $c_1^*$ and $c_2^*$. If both vertices have an edge to $B_1\cup A \cup B_2$, the vertices are at distance at most 4. Hence, one of the vertices must not be adjacent to any of $B_1\cup A \cup B_2$, but instead incident on an edge in $G[C]$. Let this vertex be $c_1^*$ with a neighbour $c_1'\in C$. By Claim~\ref{clm:2P2+P1_edgeCattachedB} we get that $c_1'$ is complete to $B_1\cup B_2$. But then $c_2^*$ must be non-adjacent to $B_1\cup B_2$, and not incident to any edge in $G[C]$ because of Claim~\ref{clm:2P2+P1_edgeCattachedB}. Hence, $c_2^*$ must have its neighbourhood contained in $A$.

		Note that the algorithm searches for $x_3\in C$ of minimum degree with $N(x_3)\subseteq A$.
		Assume for sake of contradiction that the algorithm does not find a length-5 shortest path from the BFS from $x_3$, but a length-5 shortest path does exist. By the above analysis, one of the endpoints of the length-5 diametral path must be a vertex in $c_2^* \in C$ with $N(c_2^*)\subseteq A$, and the other endpoint must be a vertex $c_1^*\in C$ with $N(c_1^*) \cap (B_1\cup A\cup B_2) = \emptyset$. Hence, the following path exists (and is a diametral path): $\langle c_1^*, c_1', b^*, u, a^*, c_2^*\rangle$ with $c_1'\in C$, $b^*\in B$ and $a^*\in A$. Then $c_2^*$ cannot be a false twin of $x_3$, as otherwise the BFS would find a length-5 shortest path. The same holds if $N(x_3)\subseteq N(c_2^*)$, so $x_3$ has a private neighbour in $A$ with respect to $c_2^*$. But, $x_3$ was chosen as a vertex of minimum degree among those $c\in C$ with $N(c) \subseteq A$, so $c_2^*$ must also have a private neighbour in $A$ with respect to $x_3$. Let $a'\in A$ denote the private neighbour of $c_2^*$ with respect to $x_3$, so $(c_2^*,a')\in E$ and $(x_3,a')\notin E$. Because the distance from $c_2^*$ to $c_1^*$ is 5, it must be that $a'$ is non-adjacent to $c_1'$ and $c_1^*$. But then, $x_3 + \langle a',u \rangle + \langle c_1',c_1^* \rangle$ forms a $2P_2+P_1$, a contradiction.

		We can conclude that finding a vertex $x_3\in C$ with minimum degree among those vertices $c\in C$ with $N(c)\subseteq A$ and executing a BFS from it suffices to decide whether a length-5 diametral path of this case exists.
		
		\medskip\noindent
		To conclude, we find a length-5 diametral path if and only if it exists in $G$.
	\end{proof}

	\subsection{Proofs of Theorem~\ref{thm:diamalggeneraloverview} and~\ref{thm:AlgOverview}}\label{sec:overviewproofs}

	We now prove our algorithmic theorems, Theorem~\ref{thm:diamalggeneraloverview} and Theorem~\ref{thm:AlgOverview}.

	\DiamAlgGeneralOverview*
	\begin{proof}
		Theorem~\ref{thm:2P1+P2free} shows that we can compute the diameter of an $H$-free graph in linear time when $H\subseteq_i P_2+2P_1$, and Theorem~\ref{thm:P3+P1free} shows that we can compute the diameter of an $H$-free graph in linear time when $H\subseteq_i P_3 + P_1$. The diameter of any $H$-free graph with $H\subseteq_i P_4$ is at most 2, so we can decide on the diameter of any such graph in linear time by checking if the graph is a clique.
	\end{proof}

	\AlgOverview*
	\begin{proof}
		The cases of $H \subseteq_i P_2 + 2P_1$, $H \subseteq_i P_3 + P_1$, and $H\subseteq_i P_4$ are all given by Theorem~\ref{thm:diamalggeneraloverview}.
		The case of $H = 4P_1$ is shown by Corollary~\ref{cor:4P1free}.
		The case of $H = P_2 + 3P_1$ is shown by Theorem~\ref{thm:P2+3P1free}.
		The case of $H = 2P_2 + P_1$ is given by Theorem~\ref{thm:2P2+P1free}.
		$H = P_3 + 2P_1$ is given by Theorem~\ref{thm:P3+2P1free}.
		When $H = P_4 + P_1$ we get a linear-time algorithm by Theorem~\ref{thm:P4+P1free}.
		This completes the proof of the theorem.
	\end{proof}
	
	\section{Conclusion and Discussion}\label{sec:conclusion}
	
	We analysed the complexity of computing the diameter of $H$-free graphs. For several $H$-free graph classes $\mathcal{G}$ where $H$ is a linear forest, we found linear time algorithms for solving \dmaxproblem{}. For some such graph classes, hardness results may exist for computing entirely the diameter (note the particular case of $H = 2P_2 + P_1$). However, for $H=2P_2$ and $H=P_t$ for odd $t\geq 5$, we show that any such a linear-time algorithm would refute SETH\@. We conjecture that this pattern on $P_t$-free graphs extends to all $t\geq 5$, as it seems a mere technical limitation that Theorem~\ref{thm:Ptodd} does not hold for even $t$, not a barrier for hardness.

	\begin{restatable}{conjecture}{ConjPtHard}\label{conj:Ptfree}
		Let $\mathcal{G}$ be the class of $P_t$-free graphs with $t\geq 5$. 
		Under SETH, \dmaxproblem{} cannot be solved in $O(n^{2-\epsilon})$ time, for all $\epsilon > 0$.
	\end{restatable}

	The hardness implied by Conjecture~\ref{conj:Ptfree} is in stark contrast to our positive algorithmic results, summarized by Theorem~\ref{thm:AlgOverview}. These results suggest that the \dmaxproblem{} problem may be easier than deciding on the diameter of an $H$-free graph in general. In an attempt to reveal the underlying pattern, we conjecture our algorithmic results generalize to broad classes of $H$-free graphs.
	
	\begin{restatable}{conjecture}{ConjPaiGeneral}\label{conj:Paigeneral}
		For a row of integers ${(a_i)}_{i=1}^k$, $k\geq 2$, ${(a_i)}_{i=1}^k\neq (2,2)$, let $\mathcal{G}$ be the class of $\sum_{i=1}^{k}P_{a_i}$-free graphs. Then \dmaxproblem{} can be solved in $O(n+m)$ time.
	\end{restatable}

	Conjecture~\ref{conj:Paigeneral} rules out $k=1$ as these cases are hard under Theorem~\ref{thm:Ptodd} and Conjecture~\ref{conj:Ptfree}. It also rules out the case of $2P_2$-free graphs, for which we showed that deciding whether the diameter is equal to $d_{\max} = 3$ is hard under SETH (see Theorem~\ref{thm:Ptodd}).
	
	Because of the generic nature of the statement in Conjecture~\ref{conj:Paigeneral}, it may very well turn out the statement is false. We give a more specific version of Conjecture~\ref{conj:Paigeneral}, which one may find more plausible. In particular, this version of the conjecture lays emphasis on linear forests $H$ including some set of isolated vertices, as it seems that in our algorithms, the presence of a $P_1$ in $H$ helps out in structural analysis.
	
	\begin{restatable}{conjecture}{ConjPtP1}\label{conj:rPt+sP1}
		Let $\mathcal{G}$ be the class of $(rP_t + sP_1)$-free graphs, where $r,s,t\geq 1$. Then \dmaxproblem{} can be solved in $O(n+m)$ time.
	\end{restatable}
	
	Conjecture~\ref{conj:rPt+sP1} demands $s \geq 1$, as $s = 0$ would contradict Theorem~\ref{thm:Ptodd}. $t\geq 1$ must hold for the correctness of the statement of $d_{\max}$.

	Note that none of our results distinguish between Conjecture~\ref{conj:Paigeneral} and Conjecture~\ref{conj:rPt+sP1}. That is, all our results support the more specific Conjecture~\ref{conj:rPt+sP1}, which in turn supports the more general conjecture. To the best of our knowledge, there is no evidence against either conjecture. A linear-time algorithm for the only open case with $d_{\max} = 4$, which is for $H=P_3+P_2$, would support Conjecture~\ref{conj:Paigeneral} but not Conjecture~\ref{conj:rPt+sP1}, and would suggest Conjecture~\ref{conj:rPt+sP1} is not the full truth.
	
	The main open problem posed by our work is then whether the Conjectures~\ref{conj:Ptfree},~\ref{conj:Paigeneral}, and/or~\ref{conj:rPt+sP1} hold true. The smallest open cases for algorithmic results to support Conjectures~\ref{conj:Paigeneral} and~\ref{conj:rPt+sP1} would be the class of $H$-free graphs with one of $H = 5P_1$ ($d_{\max} = 7$), $H = 2P_2+2P_1$ ($d_{\max} = 7$), $H = P_3+P_2$ ($d_{\max} = 4$), $H = P_4+2P_1$ ($d_{\max} = 6$), or $H = P_5+P_1$ ($d_{\max} = 5$). The smallest open case for a hardness result to support Conjecture~\ref{conj:Ptfree} is the class of $P_6$-free graphs. Conversely, an algorithmic result for $P_6$-free graphs refutes Conjecture~\ref{conj:Ptfree}, and hardness results for the other classes above refute Conjecture~\ref{conj:Paigeneral} or Conjecture~\ref{conj:rPt+sP1}.
	The specific case of $H = P_4+2P_1$ is interesting because it is the only graph $H$ for which we have no result.

	Further open questions are revealed by our results, or in particular, what our results do not show. For instance, for $(2P_2+P_1)$-free graphs, we know that we can decide in linear time whether the diameter is equal to 5, and cannot decide in truly subquadratic time whether the diameter is 2 or 3. So, settling the following open question would form a complete dichotomy for $(2P_2+P_1)$-free graphs:
	\begin{openquestion}\label{oq:2p2+p1}
		Given a $(2P_2+P_1)$-free graph $G$, can we decide in $O(n+m)$ time whether the diameter of $G$ is equal to $4$?
	\end{openquestion}
	
	We make progress for {\sc Diameter} computation on $H$-free graphs, but do not settle its complexity completely, so we ask:
	\begin{openquestion}
		When $H=4P_1$, $H=P_2+3P_1$, $H = P_3+2P_1$, $H=P_4+2P_1$, or $H = P_4+P_1$, can we solve the {\sc Diameter} problem on connected $H$-free graphs in $O(n+m)$ time?
	\end{openquestion}

	\bibliographystyle{plain}
	\bibliography{bib}
	
	\appendix
	
	\section{Existing Results}\label{appendix}
	
	\begin{figure}
		\centering
		\includegraphics[width=.8\textwidth]{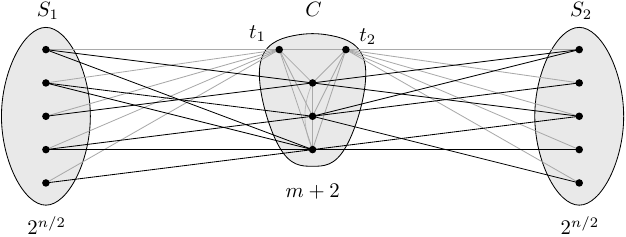}
		\caption{Reduction of Roditty and Williams~\cite{RodittyWilliams13}. The vertices corresponding to assignments of variables in the sets $S_1$ and $S_2$ form an independent set, and the $m+2$ vertices of $C$ form a clique. Hence, the graph is a split graph.}\label{fig:RWReduc}
	\end{figure}

	The following theorem and proof are due to Roditty and Williams~\cite{RodittyWilliams13}. We provide this result here because we argue about augmented versions of this construction several times.

	\begin{theorem}[\cite{RodittyWilliams13}]\label{thm:RodWilHardness}
		Under SETH, one cannot distinguish between diameter 2 and 3 in an $M$-edge split graph in $O(M^{2-\epsilon})$ time for any constant $\epsilon > 0$.
	\end{theorem}
	\begin{proof}
		Consider an instance of \textsc{CNF-SAT} on $n$ variables and $m$ clauses. Let $X_1$ and $X_2$ be a partitioning of the $n$ variables, where $|X_1| = |X_2| = n/2$. For $1\leq i \leq 2$, create a set $S_i$ of vertices containing a vertex for every assignment to the variables in $X_i$. We get $|S_1| = |S_2| = 2^{n/2}$. For $1\leq i \leq 2$, create a vertex $t_i$ that is complete to $S_i$. Create a set $C$ of vertices that contains a vertex for each clause, and also contains $t_1$ and $t_2$. Make $C$ into a clique; note that it has size $m+2$. Now, connect a vertex in $S_1$ or $S_2$ to a vertex in $C\setminus \{t_1,t_2\}$ if the corresponding assignment does not satisfy the corresponding clause. Figure~\ref{fig:RWReduc} contains an illustration of the construction. Note that the construction is a split graph, because $C$ is a clique and $S_1 \cup S_2$ is an independent set, and the graph has $O(m + 2^{n/2})$ nodes and $O(m2^{n/2})$ edges.

		We claim that the diameter of this graph is 3 if and only if there is a satisfying assignment to the instance of \textsc{CNF-SAT}. To see this, note that the diameter is at least 2 and at most 3, and a distance-3 pair must have one vertex in $S_1$ and one in $S_2$. Any two assignments that both do not satisfy some clause $c \in C$ have $c$ as a common neighbour. If there are two partial assignments, $\phi_1 \in S_1$ and $\phi_2 \in S_2$, that together form a satisfying assignment, then the distance from $\phi_1$ to $\phi_2$ is at least 3 because they do not have a common neighbour. Hence, any algorithm for \textsc{Diameter} that can distinguish between diameter 2 and 3 in an $M$-edge split graph in time $O({M}^{2-\epsilon})$ for some constant $\epsilon > 0$, would imply an algorithm for \textsc{CNF-SAT} in $O\left({m2^{(n/2)}}^{(2-\epsilon)}\right) = O\left(2^{n(1-\epsilon/2)}\mathrm{poly}(n,m)\right)$ time, and refute SETH\@.
	\end{proof}

	Related to the above theorem is the {\sc Orthogonal Vectors} problem, where, for a value $d = \omega(\log n)$, we are given two sets $A,B\subseteq {\{0,1\}}^d$, $|A| = |B| = n$, and are tasked to determine whether there are $a\in A$, $b\in B$ with $a\cdot b = 0$, that is, find two orthogonal vectors. Under SETH, we cannot solve {\sc Orthogonal Vectors} in truly subquadratic time~\cite{Williams05}. A reduction with the same hardness implication for {\sc Diameter} can be done from {\sc Orthogonal Vectors}~\cite{AbboudWW16,Williams19OnSome}, and this reduction is essentially equivalent to that of Theorem~\ref{thm:RodWilHardness}. To see this, let $S_1$ have a vertex per vector in $A$, and $S_2$ a vertex per vector in $B$. Let $C$ have $d+2$ vertices. Again, $C$ is a clique, and has two vertices $t_1,t_2$ where $t_1$ is complete to $S_1$ and $t_2$ is complete to $S_2$. Then, connect a vertex in $S_1$ or $S_2$ with a vertex in $C$ if the corresponding vector has a $1$ at the corresponding index. Now, the diameter of the graph is 3 if and only if there is an orthogonal pair of vectors.

	When discussing $4P_1$-free graphs in the introduction, we mentioned a possible approach where we make this construction into three cliques. And, a linear-time algorithm for $4P_1$-free graphs would imply an $O(n^2 + d^2)$ algorithm for {\sc Orthogonal Vectors}. To see this, assume we have such an algorithm, and take some input to {\sc Orthogonal Vectors} $A,B$. As above, create the construction that reduces {\sc Orthogonal Vectors} to {\sc Diameter}, and also make $S_1$ and $S_2$ into a clique. The graph has $O(n+d)$ vertices and $O(n^2 + d^2)$ edges, and is clearly $4P_1$-free. It takes $O(n^2 + d^2)$ time to construct the graph. But then a linear-time algorithm for {\sc Diameter} on $4P_1$-free graphs implies an algorithm for {\sc Orthogonal Vectors} in $O(n^2 + d^2)$ time.

\end{document}